\documentclass[a4paper,twocolumn,11pt,accepted=2023-05-09]{quantumarticle}
\pdfoutput=1

\usepackage[T1]{fontenc}
\usepackage{amsmath,amsfonts,amssymb,amsthm,bbm,bm, braket}
\usepackage{wasysym}
\usepackage{comment}
\usepackage{scalerel}
\usepackage{enumerate}
\usepackage{graphicx}
\usepackage{mathtools}
\usepackage{ifthen}
\usepackage{tensor}
\usepackage{tikz}
\usepackage{tikz-network}
\usetikzlibrary{patterns,decorations.pathreplacing}

\usepackage{color}
\usepackage{longtable}
\usepackage[normalem]{ulem} 

\theoremstyle{break}        
\newtheorem{Proposition}{Proposition}
\usepackage{color}
\definecolor{myred}{RGB}{232,102,102}
\definecolor{myblue}{RGB}{187,187,255}
\definecolor{myorange0}{RGB}{252,226,5}
\definecolor{myorange0c}{RGB}{255,255,255}
\definecolor{myorange}{RGB}{255,165,0}
\definecolor{mygrey}{RGB}{105,105,105}
\definecolor{OliveGreen}{RGB}{85,107,47}
\definecolor{NavyBlue}{RGB}{0,0,128}
\definecolor{mygreen}{RGB}{34,139,34}
\definecolor{myY}{RGB}{220,255,203}
\definecolor{myYO}{RGB}{255, 220, 151}

\definecolor{mygreenc}{RGB}{150,50,50}

\usepackage{tensor}
\newcommand{\be}{\begin{equation}}
\newcommand{\ee}{\end{equation}}
\newcommand{\ba}{\begin{aligned}}
\newcommand{\ea}{\end{aligned}}
\newcommand{\bw}{\begin{widetext}}
\newcommand{\ew}{\end{widetext}}
\newcommand{\1}{\mathbbm{1}}

\theoremstyle{plain}

\theoremstyle{plain}

\theoremstyle{plain}

\usepackage[colorlinks,bookmarks=false,citecolor=NavyBlue,linkcolor=OliveGreen,urlcolor=blue]{hyperref}

\newcommand{\Wgategreen}[2]{
\draw[very thick] (#1-0.5, #2 +0.5) -- (#1+0.5,#2-0.5);
\draw[very thick] (#1-0.5,#2-0.5) -- (#1+0.5,#2+0.5);
\draw[ thick, fill=mygreen, rounded corners=2pt] (#1-0.25,#2+0.25) rectangle (#1+0.25,#2-0.25);
\draw[thick] (#1,#2+0.15) -- (#1+0.15,#2+0.15) -- (#1+0.15,#2);
}
\newcommand{\MYcircle}[2]{
\draw[thick, fill=white] (#1,#2) circle (0.1cm); }
\newcommand{\MYsquare}[2]{
 \coordinate (Origin) at (#1,#2);
\filldraw [thick, fill=white, even odd rule] ($(Origin)+(-.1cm,-.1cm)$) coordinate (Square) -- ++(0.0cm,0.2cm) -- ++(0.2cm,0.0cm) -- ++(0.0cm,-0.2cm) -- cycle;
 }
\newcommand{\MYtriangle}[2]{
 \coordinate (Origin) at (#1,#2);
\filldraw [thick, fill=white, even odd rule] ($(Origin)+(-.0cm,{0.666*cos(60)*0.3cm})$) coordinate (Triangle) -- ++(0.15cm,-{cos(60)*0.3cm}) -- ++(-0.3cm,0.0cm) -- ++(0.15cm,{cos(60)*0.3cm}) -- cycle;
}
\newcommand{\MYcircleB}[2]{
\draw[thick, fill=black] (#1,#2) circle (0.1cm); }

\newcommand{\rhoO}[2]{
\draw[very thick] (-.5+#1,0.5+#2) -- (#1,0+#2);
\draw[very thick] (#1,0+#2) -- (.5+#1,0.5+#2);
\draw[very thick] (-.5+#1,#2) -- (.5+#1,#2);
\draw[ thick, fill=mygreen, rounded corners=2pt] (-0.35+#1,0.2-0.25+#2) rectangle (0.35+#1,0.2+0.2+#2);
\draw[very thick] (0.1+#1,0.15+.18+#2)-- (.15+0.1+#1,0.15+.18+#2) -- (.15+0.1+#1,0+.18+#2);
}

\newcommand{\mcirc}{\mathbin{\scalerel*{\fullmoon}{G}}}
\newcommand{\msquare}{\mathord{\scalerel*{\Box}{G}}}

\newcommand{\tr}{\text{tr} \, }
\newcommand{\Tr}{\text{Tr} \, }

\hypersetup{
pdftitle={Circuits of space and time quantum channels},
pdfsubject={Many-body open quantum systems},
pdfauthor={Pavel Kos, and Georgios Styliaris},
pdfkeywords={open systems, quantum circuits, quantum many-body systems}
}
\usepackage[numbers]{natbib}

\begin{document}

\title{Circuits of space and time quantum channels} 

\author{Pavel Kos}%
\orcid{0000-0002-7865-8609}
\thanks{Equal contribution.}
\author{Georgios Styliaris}%

\affiliation{%
Max-Planck-Institut f{\"{u}}r Quantenoptik, Hans-Kopfermann-Str. 1, 85748 Garching, Germany
}%
\affiliation{%
Munich Center for Quantum Science and Technology (MCQST), Schellingstr. 4, 80799 M{\"{u}}nchen, Germany
}%
\thanks{Equal contribution.}

\begin{abstract}

Exact solutions in interacting many-body systems are scarce but extremely valuable since they provide insights into the dynamics. Dual-unitary models are examples in one spatial dimension where this is possible. These brick-wall quantum circuits consist of local gates, which remain unitary not only in time, but also when interpreted as evolutions along the spatial directions. However, this setting of unitary dynamics does not directly apply to real-world systems due to their imperfect isolation, and it is thus imperative to consider the impact of noise to dual-unitary dynamics and its exact solvability.

In this work we generalise the ideas of dual-unitarity to obtain exact solutions in noisy quantum circuits, where each unitary gate is substituted by a local quantum channel. Exact solutions are obtained by demanding that the noisy gates yield a valid quantum channel not only in time, but also when interpreted as evolutions along one or both of the spatial directions and possibly backwards in time. This gives rise to new families of models that satisfy different combinations of unitality constraints along the space and time directions. We provide exact solutions for the spatio-temporal correlation functions, spatial correlations after a quantum quench, and the structure of steady states for these families of models. We show that noise unbiased around the dual-unitary family leads to exactly solvable models, even if dual-unitarity is strongly violated. We prove that any channel unital in both space and time directions can be written as an affine combination of a particular class of dual-unitary gates. Finally, we extend the definition of solvable initial states to matrix-product density operators. We completely classify them when their tensor admits a local purification.

\end{abstract}
\maketitle

\section{Introduction}

One of the key goals of studying non-equilibrium many-body quantum dynamics is to compute its correlation functions. Even local correlations alone typically contain enough information to capture the main physical behaviour of the system. In particular, they can determine the transport coefficients and can often be directly measured in a lab.

The problem is that correlations are usually impossible to compute exactly in generic many-body locally interacting systems, which is true for both closed or open dynamics (i.e., when the system is coupled to an environment). Although they are accessible in non-interacting (Gaussian) theories, the latter are, however, non-generic.
In closed systems, we have witnessed some exact results in random unitary circuits~\cite{nahum2017quantum,nahum2018operator,von2018operator,rakovszky2019sub,chan2018solution,garratt2021local}. 
Similarly to closed system dynamics, exact solutions are rare also in the realm of open quantum dynamics. Among the examples are generalisations of the non-interacting models, which include specific linear dissipation~\cite{prosen2008third}, integrable open systems~\cite{vanicat2018integrable,sa2021integrable, lei2022integrable}, spectral properties of random noisy quantum circuit~\cite{sa2020spectral}, and other specific solutions~\cite{znidaric2010exact}.

Recently, a new class of unitary models with interactions proved to have accessible correlation functions. These \emph{dual-unitary}~\cite{bertini2019exact} quantum circuits are characterised by the fact that the local gates (and full dynamics) remain unitary upon switching the roles of space and time. These models have accessible correlation functions~\cite{bertini2019exact,piroli2020exact} and can act as a starting point for perturbative calculations~\cite{kos2021correlations}. Generically they are provably chaotic~\cite{bertini2018exact,bertini2021random}, but contain also integrable and free examples. 
They have also analytically accessible entanglement entropy and operator spreading~\cite{bertini2019entanglement, gopalakrishnan2019unitary, piroli2020exact,claeys2020maximum,bertini2020scrambling}, as well as local operator entanglement entropy~\cite{bertini2020operator}. The structure of dual-unitary gates has been completely characterised for gates acting on two qubits~\cite{bertini2019exact}.
Beyond the qubit case, however, only certain non-exhaustive constructions are known~\cite{rather2020creating, gutkin2020exact, claeys2021ergodic, aravinda2021from,prosen2021manybody, marton2022remarks}, which exhibit a palette of rich behaviour.

Dual-unitary gates have also been used to prove deep thermalization through emergent state designs~\cite{ho2022exact,claeys2022emergent,ippoliti2022dynamical}, study eigenstate thermalization~\cite{fritzsch2021eigenstate},
temporal entanglement~\cite{lerose2020influence}, and aspects of their computational power have been characterised~\cite{suzuki2021computational}. They were also generalised to the case of having three  directions~\cite{jonay2021triunitary, ternary2022milbradt}, as well as utilised in connections with measurement induced phase transitions~\cite{ippoliti2021postselection, ippoliti2021fractal,lu2021spacetime}. Due to their important role, they have already been realised in experimental setups~\cite{chertkov2021holographic, mi2021information}.

Real-world experiments, especially in the era of NISQ devices~\cite{preskill2018quantum}, are never totally isolated from their surroundings. It is therefore crucial to understand how the coupling to the environment, which introduces noise and errors, influences the dynamics. As interactions are constrained by locality, quantum dynamics composed of \textit{local and noisy gates} comprise a very general framework beyond unitarity for the description of the complex dynamics exhibited by noisy quantum devices. In particular, it is imperative to examine if and how a plethora of recent results about many-body unitary circuits dynamics~\cite{nahum2017quantum,nahum2018operator,von2018operator,rakovszky2019sub,chan2018solution,garratt2021local,bertini2019exact,bertini2019entanglement, gopalakrishnan2019unitary, piroli2020exact,claeys2020maximum,bertini2020scrambling} are altered in the presence of local noise.
In general the open dynamics is even harder to understand than the unitary dynamics. For instance, even storing a full density matrix describing $L$ qubits requires memory scaling as $2^{2L}$ instead of $2^L$ for their wave function.

Experiments performed in~\cite{chertkov2021holographic, mi2021information}
show remarkable agreement with dual-unitary dynamics, even with inevitable imperfections. This raises several questions. For instance, does the open dynamics preserve the typical features of dual-unitarity, such as the strict maximal Lieb-Robinson velocity \cite{claeys2020maximum} and the zero correlations inside the light cone \cite{bertini2019exact,claeys2020maximum}? 
Is there some robustness to noise, which maintain dual-unitary features even-though the particular gates are not exactly dual unitary? As we show, the answer is affirmative, if the errors are unbiased with respect to the dual-unitarity family.

More broadly, it is natural to wonder to what extent the exact solvability implied by dual-unitary dynamics can be adapted to the domain of open systems. First steps in this direction were made via averaging random dual-unitary evolutions~\cite{kos2021correlations,kos2021chaos}, which were shown to correspond to certain classical Markov chains. In this context, it has been pointed out that the so-called dual unitality suffices to calculate their spatiotemporal correlations exactly.

Nevertheless, a systematic identification and treatment of the counterparts of dual-unitary dynamics for \emph{open systems} is lacking. Beyond isolated examples, it is not evident what dynamical quantities of the evolution are amenable to analytical treatment, and what are the features of the evolution itself that make the resulting dynamics solvable.

In this work we establish a framework to address these questions. We first introduce distinct classes of quantum gates for open systems, which satisfy a unitality property along different combinations of space and time directions. This property generalises dual-unitarity to noisy dynamics, so we refer to such open system evolution as \textit{space-time unital}. Unlike the dual-unitary case, unitality for noisy dynamics might hold in only one of the two space directions and/or in the ordinary time direction, as well as in diagonal combinations, giving rise to different classes of dynamics.

Subsequently, we demonstrate how the introduced space-time unital models lead to exact solvability in the realm of noisy quantum circuits, generalising dual-unitarity.
In particular, we consider quantum circuits in one spatial dimension with dynamics consisting of a brick-wall architecture of space-time unital gates. We investigate the two-point spatio-temporal correlations, the two-point spatial correlations after a quench, as well as the fixed points of the dynamics. We determine the reductions that occur in the calculations of the correlation functions as a result of different kinds of space-time unitality. Importantly, we identify broad families of models resulting in exact solutions that can be obtained efficiently (in the number of time steps).

Moreover, we identify relevant instances which produce dynamics belonging in the aforementioned classes. Most importantly, we show that an unbiased averaging of \textit{non-dual-unitary} gates around dual-unitary points produces space-time unital evolutions. This provides a sufficient mechanism where noise induces solvable dynamics. It is also relevant in the context of experimental implementations of dual-unitary circuits, showing that imperfections around dual-unitarity (which inevitable occur) can lead to exact solvability, even if dual-unitarity is violated.

Solvability is a property that depends also on the structure of the initial state. For the pure state case with translational invariance, the so-called solvable initial states yield analytical tractability of the correlation functions for dual-unitary gates. Moreover, the solvable initial states have been characterised, and are essentially in one-to-one correspondence with unitary operator~\cite{piroli2020exact}. Here we extend the definition and the characterisation of solvable translation-invariant mixed states for the case of dynamics described by space-time unital quantum channels. We find that the solvable states are in one-to-one correspondence with quantum channels, under the assumption that the former admit a local purification. Moreover, we establish that, in a certain case, all solvable initial states reach a common fixed point of space-time unital evolution in a finite time.

From a more abstract quantum information perspective, the different classes of space-time unital gates forms a convex subset within the set of quantum channels. We first reformulate the different unitality constraints in terms of the corresponding Choi-Jamio{\l}kowski state and then specify the dimension for each of the introduced space-time unital gates. This allows us to establish a structural characterisation of the quantum channels that are unital in time and in both space directions; they turn out to exactly coincide with quantum channels that can be expressed as affine combinations of a particular class of dual-unitary gates.

The rest of this work is structured as follows. We first introduce the setting of space-time unital quantum circuit in Sec.~\ref{sec:setting}, where we state the analogue of dual-unitary conditions and comment on some interesting examples. We follow with section~\ref{sec:structure} about the structure of these families and section~\ref{sec:applications} on applying the conditions to obtain correlation functions. There we use solvable initial states, which we discuss in detail in section~\ref{sec:MPDO}. We finish with concluding in section~\ref{sec:conclusions}.
Some technical steps are delegated to the appendices.

\section{Setting}
\label{sec:setting}

In this work we consider one-dimensional quantum systems consisting of $2L$ qudits with $d$ internal states. 
They are positioned on sites labelled with half-integer numbers $1/2,1,\ldots L$. The system is evolved with a brick-wall quantum circuit of identical local quantum channels, that is, completely positive trace-preserving (CPTP) maps.

We will use two descriptions of quantum channels. First is the Kraus form~\cite{nielsen2011quantum,bengtsson2017geometry}:
\be
\rho(t+1) = \sum_k F_k \rho(t) F_k^\dagger \,,
\ee
whereas the second is the vectorised, or folded, description. In the latter we map the operators acting on $l$ consecutive qudits, i.e., operators over $\mathbb{C}^{d^l}$, to vectors in $\mathbb{C}^{d^l} \otimes \mathbb{C}^{d^l}$ as
\be
a \xmapsto{{\rm vec}} \ket{a}\, .
\ee
Naturally, we adopt the standard Hilbert-Schmidt scalar product
$\langle a|b\rangle = \tr a^\dagger b$ and its induced norm.
The mapping is specified by fixing a basis $\{\ket{n}\}$ of $\mathbb{C}^{d^l}$ which induces a basis $\{\ket{m} \! \bra{n}\}$ for  $\mathbb{C}^{d^l}\!\!\! \otimes \mathbb{C}^{d^l}$. Then, we define the following mapping of basis elements   
\be
\ket{m} \! \bra{n} \xmapsto{{\rm vec}}  \ket{m} \otimes \ket{n} \,,
\ee
and extend it by linearity.
The quantum channels, which acted before as super-operators, are mapped to operators acting over $\mathbb{C}^{d^{2l}}$, as seen in the Kraus form:
\be
\sum_k F_k(\cdot)F_k^\dagger \xmapsto{{\rm vec}}  \sum_k F_k \otimes F_k^* \, .
\ee
It will be convenient to denote the normalised vectorised identity operator over a local space $\mathbb C^\xi$ as  
\be
\ket{\mcirc}= \frac{1}{\sqrt{\xi}}  \ket{\1_\xi}=
\begin{tikzpicture}[baseline=(current  bounding  box.center), scale=.7]
\draw[very thick] (0,0) -- (0,0.35);
\draw[very thick, fill=white] (0,0) circle (0.1cm); 
\end{tikzpicture}\, .
\label{eq:Id}
\ee
These operators can be thought of as connecting the bras and kets.
In this notation, the trace over the whole space and infinite temperature state are
\be
\Tr_{2L} = d^L \bra{\mcirc \ldots \mcirc},
\quad 
\rho_\infty=\frac{\1}{d^{2L}}=\frac{1}{d^{L}}\ket{\mcirc \ldots \mcirc}.
\ee

We will adopt a convenient graphical representation of circuits of quantum channels, as is usually the case in tensor networks~\cite{cirac2021matrix}. As mentioned previously, we will focus on homogeneous (translationally invariant) noisy quantum circuits, i.e, quantum circuits where each gate represents a local quantum channel $q$ acting over a pair of neighbouring qudits. In particular, we will extensively employ vectorised notation and denote
\be
\label{eq:qgate}
q=
\begin{tikzpicture}[baseline=(current  bounding  box.center), scale=.7]
\Wgategreen{0}{0}\end{tikzpicture}\,  
=\, \,
\begin{tikzpicture}[baseline=(current  bounding  box.center), scale=.7]
\def\dx{0.15}
\def\dy{0.15}
\draw[ thick] (-4.25+\dx,0.5+\dy) -- (-3.25+\dx,-0.5+\dy);
\draw[ thick] (-4.25+\dx,-0.5+\dy) -- (-3.25+\dx,0.5+\dy);
\draw[ thick, fill=myred, rounded corners=2pt] (-4+\dx,0.25+\dy) rectangle (-3.5+\dx,-0.25+\dy);
\draw[ thick] (-4.25,0.5) -- (-3.25,-0.5);
\draw[ thick] (-4.25,-0.5) -- (-3.25,0.5);
\draw[ thick, fill=myblue, rounded corners=2pt] (-4,0.25) rectangle (-3.5,-0.25);
\draw[thick] (-3.75,0.15) -- (-3.75+0.15,0.15) -- (-3.75+0.15,0);
\draw[ thick] (-3.5 ,0.) -- (-3.5+.2,0.) --(-3.5+.2+\dx,0.+\dy) --(-3.5+\dx,0.+\dy);
%
\end{tikzpicture}\, \,= \, \sum_k E_k \otimes E_k^* \; .
\ee
After the second equality we used the vectorised Kraus form of the channel. We use $E_k$ to distinguish local from global Kraus operators $F_k$. Notice the additional wire connecting the two copies of the system, which represent the summation over Kraus operators.
In the graphical representation time flows from bottom upwards, i.e., input states vectors are applied at the bottom. In the same graphical notation, stacking operators amounts to tensoring them, with the box in the front corresponding to the leftmost operator in the tensor product.

One time step of the time evolution is given by the global quantum channel $\mathbb{Q}$ built out of local channels in the following way: 
\begin{align}
&\mathbb Q = \mathbb{T}_{2L}  q^{\otimes L}\mathbb{T}^\dagger_{2L}  q^{\otimes L} \\
&=
\begin{tikzpicture}[baseline=(current  bounding  box.center), scale=0.55]
\draw[thick, dotted] (-9.5,-1.7) -- (0.4,-1.7);
\draw[thick, dotted] (-9.5,-1.3) -- (0.4,-1.3);
\foreach \i in {3,...,13}{
\draw[gray, dashed] (-12.5+\i,-2.1) -- (-12.5+\i,0.3);
}
\foreach \i in {3,...,5}{
\draw[gray, dashed] (-9.75,3-\i) -- (.75,3-\i);
}
\foreach \i in{0.5}
{
\draw[ thick] (0.5,1+2*\i-0.5-3.5) arc (-45:90:0.17);
\draw[ thick] (-10+0.5,1+2*\i-0.5-3.5) arc (270:90:0.15);
}
\foreach \i in{1.5}{
\draw[thick] (0.5,2*\i-0.5-3.5) arc (45:-90:0.17);
\draw[thick] (-10+0.5+0,2*\i-0.5-3.5) arc (90:270:0.15);
}
\foreach \i in {0,1}{
\Text[x=1.25,y=-2+2*\i]{\scriptsize$\i$}
}
\foreach \i in {1}{
\Text[x=1.25,y=-2+\i]{\small$\frac{\i}{2}$}
}
\foreach \i in {1,3,5}{
\Text[x=-7.5+\i+1-3,y=-2.6]{\small$\frac{\i}{2}$}
}
\foreach \i in {1,2,3}{
\Text[x=-7.5+2*\i-2,y=-2.6]{\scriptsize${\i}$}
}
\Text[x=-7.5+2*5-2,y=-2.6]{\small $L$}
\foreach \jj[evaluate=\jj as \j using -2*(ceil(\jj/2)-\jj/2)] in {0}
\foreach \i in {1,...,5}
{
\draw[thick] (.5-2*\i-1*\j,-2-1*\jj) -- (1-2*\i-1*\j,-1.5-\jj);
\draw[thick] (1-2*\i-1*\j,-1.5-1*\jj) -- (1.5-2*\i-1*\j,-2-\jj);
\draw[thick] (.5-2*\i-1*\j,-1-1*\jj) -- (1-2*\i-1*\j,-1.5-\jj);
\draw[thick] (1-2*\i-1*\j,-1.5-1*\jj) -- (1.5-2*\i-1*\j,-1-\jj);
\draw[thick, fill=mygreen, rounded corners=2pt] (0.75-2*\i-1*\j,-1.75-\jj) rectangle (1.25-2*\i-1*\j,-1.25-\jj);
\draw[thick] (-2*\i+1,-1.35-\jj) -- (-2*\i+1.15,-1.35-\jj) -- (-2*\i+1.15,-1.5-\jj);%
}
\foreach \i in {0,...,4}{
\draw[thick] (-.5-2*\i,-1) -- (0.525-2*\i,0.025);
\draw[thick] (-0.525-2*\i,0.025) -- (0.5-2*\i,-1);
\draw[thick, fill=mygreen, rounded corners=2pt] (-0.25-2*\i,-0.25) rectangle (.25-2*\i,-0.75);
\draw[thick] (-2*\i,-0.35) -- (-2*\i+0.15,-0.35) -- (-2*\i+0.15,-0.5);%
}
\Text[x=-2,y=-2.6]{$\cdots$}
\end{tikzpicture}\, ,
\label{eq:Q}
\end{align}
where $\mathbb{T}_{2L}$ is an operator implementing $2L$ periodic translation by one site. We will focus on two kinds of boundary conditions: homogenous circuits with periodic boundary conditions or open boundary conditions with fully depolarising quantum channels at the edges.

Our restriction to the homogeneous setting is made for convenience and most of the subsequent discussion carries over to the inhomogenous case.
Next, we discuss possible additional constraints on the local quantum channels $q$. Different combinations of them produce different families of models with varying degree of analytical solvability.

\subsection{Trace preservation and unital quantum channels}

Any valid quantum channel $q$ preserves the trace, which is reflected in a condition on the local gate as
\begin{align}
&\textbf{(i) Trace Preservation (TP):} \quad 
\sum_k E_k^\dagger E_k = \1\;,  
\nonumber\\ &\bra{\mcirc \mcirc} q =
\bra{\mcirc \mcirc}\;,
\quad
\begin{tikzpicture}[baseline=(current  bounding  box.center), scale=.7]
\Wgategreen{0}{0}
\draw[thick, fill=white] (0-0.5,0.5) circle (0.1cm); 
\draw[thick, fill=white] (0.5,0.5) circle (0.1cm); 
\end{tikzpicture}\,  =
\begin{tikzpicture}[baseline=(current  bounding  box.center), scale=.7]
\draw[thick](-.5,-.5)--(-.5,.5);\draw[thick](.5,-.5)--(.5,.5);
\draw[thick, fill=white] (0-0.5,0.5) circle (0.1cm); 
\draw[thick, fill=white] (0.5,0.5) circle (0.1cm); \end{tikzpicture}\;
,
\end{align}
where the three expressions are equivalent. This channel is completely positive by construction, since it is written in the Kraus form.

If the local quantum channel maps $\1$ to $\1$, then the channel is
\begin{align}
&\textbf{(ii) Time Unital (TU):} \quad  
\sum_k E_k E_k^\dagger  = \1\;, 
\nonumber\\ 
&q  \ket{\mcirc \mcirc}= \ket{\mcirc \mcirc}\;, 
\quad
\begin{tikzpicture}[baseline=(current  bounding  box.center), scale=.7]
\Wgategreen{0}{0}
\draw[thick, fill=white] (0-0.5,0-0.5) circle (0.1cm); 
\draw[thick, fill=white] (0.5,0-0.5) circle (0.1cm); 
\end{tikzpicture}\,  =
\begin{tikzpicture}[baseline=(current  bounding  box.center), scale=.7]
\draw[thick](-.5,-.5)--(-.5,.5);\draw[thick](.5,-.5)--(.5,.5);
\draw[thick, fill=white] (0-0.5,0-0.5) circle (0.1cm); 
\draw[thick, fill=white] (0.5,0-0.5) circle (0.1cm); \end{tikzpicture}\, .
\end{align}
This also implies the unitality of the global channel $\mathbb{Q}$. Evidently, trace preservation amounts to unitality in the backward-time direction. This is since reversing the time direction (i.e., interpreting the graphical notation from top to bottom), trace preservation turns into (time) unitality and vice-versa.

\subsection{Space quantum channels}

The defining feature of dual-unitary gates~\cite{bertini2019exact} is their property to maintain unitarity even when they are interpreted as evolutions along the spatial directions. In a natural extension of dual-unitarity to quantum channels, we demand that the evolution along one or both of the the space directions yields a physical evolution, i.e., a legitimate quantum channel. Let us first demand that it is a space quantum channel when considered from right to left. As in the folded picture operators are by construction in the Kraus form, complete positivity is automatic and we only need to ensure trace preservation.

For that, we define the space-time rotation operation $\sim$ which maps the right (left) 
indices to the inputs (outputs)
\be
\label{eq:tildeqgate}
\tilde{q}=\begin{tikzpicture}[baseline=(current  bounding  box.center), scale=.7]
\draw[ thick] (-4.25,0.5) -- (-3.25,-0.5);
\draw[ thick] (-4.75,0.5)-- (-4.75,-0.5)--(-4.25,-0.5) -- (-3.25,0.5)-- (-2.75,0.5) -- (-2.75, -0.5);
\draw[ thick, fill=mygreen, rounded corners=2pt] (-4,0.25) rectangle (-3.5,-0.25);
\draw[thick] (-3.75,0.15) -- (-3.75+0.15,0.15) -- (-3.75+0.15,0);
\Text[x=-4.25,y=-0.75]{}
\end{tikzpicture}\; .
\ee
Therefore trace preservation of $\tilde{q}$ ensures that the evolution is a legitimate quantum channel, when considered from right to left. Equivalently, this corresponds to unitality in the left-to-right space direction. As such, we say that $q$ is
\begin{align}
\label{eq:leftU}
&\textbf{(iii) Left Unital (LU):} \quad
 \sum_k (\tilde{E}_k)^\dagger \tilde{E}_k = \1, \nonumber\\
 &\bra{\mcirc \mcirc} \tilde{q} = \bra{\mcirc \mcirc},
\quad
\begin{tikzpicture}[baseline=(current  bounding  box.center), scale=.7]
\Wgategreen{0}{0}
\draw[thick, fill=white] (0-0.5,0-0.5) circle (0.1cm); 
\draw[thick, fill=white] (-0.5,0.5) circle (0.1cm); 
\end{tikzpicture}\,  =
\begin{tikzpicture}[baseline=(current  bounding  box.center), scale=.7]
\draw[thick](-.5,-.5)--(.5,-.5);\draw[thick](-.5,.5)--(.5,.5);
\draw[thick, fill=white] (0-0.5,0-0.5) circle (0.1cm); 
\draw[thick, fill=white] (-0.5,0.5) circle (0.1cm); \end{tikzpicture}\; .
 \end{align}
Here the $\sim$ operator changes the input (output) legs of $E_k$ in the same way as for $q$.

Analogously, we may demand that the left-to-right evolution is as quantum channel, resulting in $q$ being
\begin{align}
&\textbf{(iv) Right Unital (RU):} \quad
 \sum_k  \tilde{E}_k (\tilde{E}_k)^\dagger = \1, \nonumber
\\ &\tilde{q} \ket{\mcirc \mcirc}  = \ket{\mcirc \mcirc},\quad
\begin{tikzpicture}[baseline=(current  bounding  box.center), scale=.7]
\Wgategreen{0}{0}
\draw[thick, fill=white] (0.5,0-0.5) circle (0.1cm); 
\draw[thick, fill=white] (0.5,0.5) circle (0.1cm); 
\end{tikzpicture}\,  =
\begin{tikzpicture}[baseline=(current  bounding  box.center), scale=.7]
\draw[thick](-.5,-.5)--(.5,-.5);\draw[thick](-.5,.5)--(.5,.5);
\draw[thick, fill=white] (0.5,0-0.5) circle (0.1cm); 
\draw[thick, fill=white] (0.5,0.5) circle (0.1cm); \end{tikzpicture}
\; .
\end{align}

\begin{table*}
    \centering
    \begin{tabular}{c|c|c}
         \textbf{Class} & \textbf{Conditions} &\textbf{Alternative Conds.} \\
        \hline
        2-way unital & (i), (iii)   & (i), (iv)
          \\
     3-way time unital & (i), (ii), (iii)   & (i), (ii), (iv)
      \\
     3-way space unital & (i), (iii), (iv)   & 
      \\
     4-way unital & (i), (ii), (iii), (iv)   &
    \end{tabular}
    \caption{
    A summary of different classes and the conditions which they satisfy.
    }
    \label{tab:Classes}
\end{table*}

In the following, we consider families of models which satisfy different combinations of the aforementioned conditions, which we summarise in Tab. \ref{tab:Classes}. Due to the physical relevance of quantum channels, we henceforth always assume property (i). If, in addition, one of the (iii)-(iv) holds, we will refer to such quantum channels as \textit{2-way unital}, and similarly, for two conditions out of (ii)-(iv), we will refer to them as \textit{3-way unital}. In this class, there are two possibilities with different physical consequences. The first family is \textit{3-way time unital} channels, which satisfies conditions (i)-(iii) [or (i), (ii), and (iv)], and the second family is \textit{3-way space unital} channels, which satisfies conditions (i), (iii), and (iv). The last possibility is when all conditions (i)-(iv) are satisfied, in which case we call the quantum channel \textit{4-way unital}. We will momentarily show that this is indeed possible, i.e., there exist nontrivial representatives belonging strictly to each of the above classes.

Let us stress that convex combinations (averages) of dual-unitary gates will always yield a subset of 4-way unital examples, and other families cannot be achieved in this way. However, as we will subsequently show, 3-way unitality turns out to be sufficient for exact solvability.

\subsection{Perfect quantum channels}

Perfect tensors were introduced in~\cite{pastawski2015holographic} as the building blocks of exactly solvable toy models for the AdS/CFT correspondence, and also independently in~\cite{goyeneche2015absolutely} where the corresponding matrices were called two-unitary.
They are tensors with an even number of indices whose defining property is that any balanced bipartition of these indices into inputs and outputs gives a unitary transformation\footnote{As such, they are closely connected with absolutely maximally entangled states~\cite{goyeneche2015absolutely}, i.e., multipartite states that have maximally mixed marginals after tracing out at least half of the system. This connection can be anticipated by recalling that every bipartite maximally entangled states is in correspondence with a vectorised unitary (since $\Tr_B ( \ket{U}_{AB} \! \bra{U}) = \1$)~\cite{bengtsson2017geometry}.
In the case of four indices, the perfect tensor conditions are equivalent to unitarity, dual-unitarity and unitarity of the partial transpose  (T-duality)~\cite{aravinda2021from}. The last condition makes the class more restrictive than dual-unitarity. Note that the matrices corresponding to perfect tensors are also called 2-unitaries~\cite{goyeneche2015absolutely}, which are not to be confused with dual-unitary matrices.
}.

In analogy to perfect tensors, we define perfect channels. Not only do they satisfy conditions (i)-(iv), but also the diagonal unitality conditions: 
\begin{align}\begin{split}
\! \! \! \! \! &\textbf{(v):} \;
\sum_k  \widetilde{( E_k S)} \widetilde{(S E_k^\dagger)} = \1, \;
\begin{tikzpicture}[baseline=(current  bounding  box.center), scale=.7]
\Wgategreen{0}{0}
\draw[thick, fill=white] (0.5,0.5) circle (0.1cm); 
\draw[thick, fill=white] (-0.5,-0.5) circle (0.1cm); 
\end{tikzpicture}\,  =
\begin{tikzpicture}[baseline=(current  bounding  box.center), scale=.7]
\draw[thick](-.3,0.3)--(-.6,.6);
\draw[thick](.3,-.3)--(.6,-.6);
\draw[thick, fill=white] (0.3,0-0.3) circle (0.1cm); 
\draw[thick, fill=white] (-0.3,0.3) circle (0.1cm); \end{tikzpicture}\; ,
 \end{split} \\
\begin{split}
&\textbf{(vi):} \;
\sum_k  \widetilde{(S E_k^\dagger)} \widetilde{( E_k S)}  = \1, \;
\begin{tikzpicture}[baseline=(current  bounding  box.center), scale=.7]
\Wgategreen{0}{0}
\draw[thick, fill=white] (0.5,-0.5) circle (0.1cm); 
\draw[thick, fill=white] (-0.5,0.5) circle (0.1cm); 
\end{tikzpicture}\,  =
\begin{tikzpicture}[baseline=(current  bounding  box.center), scale=.7]
\draw[thick](.3,0.3)--(.6,.6);
\draw[thick](-.3,-.3)--(-.6,-.6);
\draw[thick, fill=white] (0.3,0.3) circle (0.1cm); 
\draw[thick, fill=white] (-0.3,-0.3) circle (0.1cm); \end{tikzpicture}
\; ,
 \end{split}
\end{align}
where $S$ is the swap operator.
The simplest example of a perfect channel is the completely depolarising map
\begin{align*}
    \mathcal T (\cdot) = \Tr(\cdot) \frac{\1}{d^2} \; .
\end{align*}
It fulfils all conditions (i)-(vi), as is apparent from its folded graphical representation:
\begin{align}
    \begin{tikzpicture}[baseline=(current  bounding  box.center), scale=.7]
\Wgategreen{0}{0}
\end{tikzpicture}\;  = 
\begin{tikzpicture}[baseline=(current  bounding  box.center), scale=.7]
\draw[thick](.3,0.3)--(.6,.6);
\draw[thick](-0.3,0.3)--(-0.6,.6);
\draw[thick](0.3,-0.3)--(0.6,-.6);
\draw[thick](-.3,-.3)--(-.6,-.6);
\draw[thick, fill=white] (0.3,0.3) circle (0.1cm); 
\draw[thick, fill=white] (-0.3,0.3) circle (0.1cm); 
\draw[thick, fill=white] (0.3,-0.3) circle (0.1cm); 
\draw[thick, fill=white] (-0.3,-0.3) circle (0.1cm); \end{tikzpicture} \,.
\end{align}
Remarkably, in contrast to perfect tensors which exist only for $d\geq3$~\cite{aravinda2021from}, perfect channels exist already for qubits ($d=2$).
A systematic way of constructing them will follow in the next section.

\section{Structure of Space-Time Unital Quantum Channels}
\label{sec:structure}
Here we investigate further the general structure of the previously introduced families of models. Specifically, after reformulating the constraints (i)-(vi) through the channel-state duality, we determine the dimensions of the corresponding manifolds of quantum channels by counting the number of ways we can deform the completely depolarising channel. This, together with convexity, is enough to give us the dimensions of these spaces. We also show that all 4-way unital channels can be decomposed as an affine combination of dual-unitary gates, in analogy with a matching result for ordinary (time) unital channels and unitary gates.

\subsection{The Choi-Jamio{\l}kowski state of space-time unital channels}

\begin{table}[t]
    \centering
    \begin{tabular}{c|c}
                            & \textbf{Maximally-mixed} \\
         \textbf{Condition} & \textbf{marginal} \\
        \hline
        \textbf{(i)} Trace Preservation      & $\Tr_{\!AB} \left[ J(\mathcal E_{AB})  \right] $  \\
        \textbf{(ii)} Time Unital     &  $\Tr_{\!A'B'} \left[ J(\mathcal E_{AB})  \right] $  \\
        \textbf{(iii)} Left Unital    &  $\Tr_{\!AA'} \left[ J(\mathcal E_{AB})  \right] $ \\
        \textbf{(iv)} Right Unital    &  $\Tr_{\!BB'} \left[ J(\mathcal E_{AB})  \right] $  \\
        \textbf{(v)} Diagonal         &  $\Tr_{\!A'B} \left[ J(\mathcal E_{AB})  \right] $ \\
        \textbf{(vi)} Diagonal        &  $\Tr_{\!AB'} \left[ J(\mathcal E_{AB})  \right] $ \\
    \end{tabular}
    \caption{Reformulation of the space-time unitality conditions (i)-(vi) in terms of the Choi-Jamio{\l}kowski state~\eqref{eq:Choi}. The condition stated on the left column is equivalent to the corresponding marginal being maximally mixed.
    }
    \label{tab:Choi}
\end{table}

It will be beneficial to reformulate the space-time unitality conditions (i)-(vi) in terms of the corresponding Choi-Jamio{\l}kowski state (CJS). The utility of this representation arises from the fact that a superoperator $\mathcal E_{AB}$, acting on regions $A$ and $B$, is completely positive if and only if the corresponding CJS
\begin{align} \label{eq:Choi}
    J(\mathcal E_{AB}) = \mathcal E_{AB} \otimes \1_{A'B'} (\ket{\Phi}_{ABA'B'} \! \bra{\Phi} )
\end{align}
is positive~\cite{nielsen2011quantum}, where $\textstyle \ket{\Phi}_{ABA'B'} = \frac{1}{d} \textstyle \sum_{ij = 1}^d \ket{ijij}$.

In this representation, each of the unitality conditions results in an equivalent constraint over the CJS, imposing that an appropriate marginal ought to be maximally mixed. The correspondence is summarised in Table~\ref{tab:Choi}. The derivation can be conveniently carried out pictorially in the folded notation, in which the CJS takes the form
\begin{align} \label{eq:choi_graphical}
 \ket{J(\mathcal E_{AB})} = 
    \begin{tikzpicture}[baseline=(current  bounding  box.center), scale=.7]
\Wgategreen{0}{0}
\draw[very thick](-0.5,-0.5) to[out=-135, in=-90] (1.25,-.50)--(1.25,.5);
\draw[very thick](0.5,-0.5) to[out=-45, in=-90] (2.25,-.50)--(2.25,.5);
\Text[x=-.5,y=1.]{$A$}
\Text[x=.5,y=1.]{$B$}
\Text[x=1.25,y=1.]{$A'$}
\Text[x=2.25,y=1.]{$B'$}
\end{tikzpicture}\; .
\end{align}

For example, the left-unital condition gives the maximally-mixed marginal $\Tr_{\!AA'} \left[ J(\mathcal E_{AB})  \right] = \1 $. This is most easily verified in vectorised graphical notation
\be
\begin{tikzpicture}[baseline=(current  bounding  box.center), scale=.7]
\Wgategreen{0}{0}
\draw[very thick](-0.5,-0.5) to[out=-135, in=-90] (1.25,-.50)--(1.25,.5);
\draw[very thick](0.5,-0.5) to[out=-45, in=-90] (2.25,-.50)--(2.25,.5);
\Text[x=-.5,y=1.]{$A$}
\Text[x=.5,y=1.]{$B$}
\Text[x=1.25,y=1.]{$A'$}
\Text[x=2.25,y=1.]{$B'$}
\draw[thick, fill=white] (-0.5, 0.5) circle (0.1cm); 
\draw[thick, fill=white] (1.25, 0.5) circle (0.1cm); 
\end{tikzpicture}\; 
=
\begin{tikzpicture}[baseline=(current  bounding  box.center), scale=.7]
\Text[x=-.5,y=1.]{$A$}
\Text[x=.5,y=1.]{$B$}
\Text[x=1.25,y=1.]{$A'$}
\Text[x=2.25,y=1.]{$B'$}
\draw[very thick](0.5,.0)--(.5,.5);
\draw[very thick]  (2.25,.0)--(2.25,.5);
\draw[thick, fill=white] (0.5, 0.) circle (0.1cm); 
\draw[thick, fill=white] (2.25, 0.) circle (0.1cm); 
\end{tikzpicture}\;,
\ee
which is equivalent to Eq.~\eqref{eq:leftU}.

\subsection{Dimension counting}\label{sec:dimension_counting}

Let us first recall that the set of quantum channels is convex~\cite{watrous2018theory}. It is then easy to see that this property is passed down to the sets of CPTP maps that satisfy any combination of the properties (ii)-(vi), as convex combinations preserve these properties. Exactly as in the case of ordinary quantum channels, each of these convex sets has a well-defined dimension corresponding to the minimum number of real parameters that are needed to describe the underlying channels. More precisely, this is the dimension of a minimal affine subspace\footnote{Affine subspace here means a subspace of the vector space that has been shifted from the origin by a constant displacement.} containing the target convex set. The underlying parent vector space can be taken to be that of hermiticity-preserving superoperators.

As we noted earlier, the completely depolarizing channel is perfect thus all of the aforementioned distinguished sets of quantum channels are non-empty. In fact, suitably ``deforming'' the completely depolarizing channel allows counting the exact dimension.

\begin{table*}[t]
    \centering
    \begin{tabular}{c|c|c|c}
        \textbf{Type} & \textbf{Conditions} & \textbf{dimension} & \textbf{co-dimension} \\
        \hline
        CPTP    & i        & $d^8- d^4$ & $0$\\
        unital  & i,ii     & $(d^4-1)^2$ & $d^4-1$\\
        2-way   & i,iii    & $d^8- 2 d^4+ d^2$ & $d^4-d^2$\\
        3-way   & i-iii    & $(d^4-1)^2- (d^2-1)^2$ & $2d^4-2d^2$\\
        4-way   & i-iv     & $(d^4-1)^2- 2(d^2-1)^2$ & $3d^4-4d^2+1$\\
        perfect & i-vi     & $(d^4-1)^2- 4(d^2-1)^2$ & $5d^4-8d^2+3$\\
    \end{tabular}
    \caption{
    Examples, dimensions and co-dimensions of different quantum channels. When various combinations of conditions are possible, we write only one possible instance.}
    \label{tab:dimensions}
\end{table*}

\begin{Proposition} \label{prop:counting_dimensions}
Consider the convex set of quantum channels acting over states in $\mathbb C^{d} \otimes \mathbb C^{d}$ (i.e., over a pair of sites) satisfying the introduced space-time unitality conditions. Then the dimensions of the different sets are as given in Table~\ref{tab:dimensions}. 
\end{Proposition}

The proof can be found in Appendix~\ref{sec:app_proofs}. It consists of counting all of the linearly independent perturbations around $\mathcal T$ that are compatible with the space-time unitality constraints. Preservation of complete positivity can then be deduced from the positivity of the CJS.

In Table~\ref{tab:dimensions} we also provide co-dimensions,  which states the number of the fixed parameters with the respect to general CPTP channels.
Note that in the limit of large local space $d \gg 1$, where dimension scales as $\sim d^8$, the co-dimension for each of the above types scales  only as $\sim r d^4$, where $r$ is the total number of conditions (ii)-(vi) that needs to be satisfied. In particular,
3-way unitality, which suffices for solvability, has $\sim d^4$ more free parameter than 4-way unitality.
All of the classes mentioned before have considerably more free parameters ($\sim d^8$) that the $\sim 2d^3$ parameters of dual-unitaries~\cite{prosen2021manybody}.

It is important to stress that our dimension-counting arguments crucially depend on convexity of quantum channels. In contrast, convexity is absent on the level of dual-unitary gates. There, the sets of dual-unitary models with fixed input-output dimensions may have a complicated structure of the underlying space, e.g., as uncovered by numerical studies in~\cite{prosen2021manybody}.

We would like to better understand the structure of the different sets of quantum channels mentioned before. One way to characterise each convex set is by starting at the completely depolarising channel, which is always contained in it. From there we investigate the distance to the boundary of the set, which is characterised by the failure of the complete positivity. 
To do that, we choose a perturbation of the completely depolarising channel respecting the correct conditions and check for what strength of perturbation complete positivity fails. The smallest such strength gives us the distance to the boundary in that specific direction. Remarkably, we numerically find that this distance is much smaller if we choose a random direction, than if we move towards a quantum channel given by a dual-unitary. 
Note that the same holds in absence of dual-unitarity, i.e., for the distance to unitary channels.
We explain these numerical experiments with more technical details in the appendix~\ref{sec:app_structure}.

%
We expect that finding the structure of the extremal points of the introduced sets of quantum channels to be challenging, since already the set of quantum channels over a single qubit is rather complicated~\cite{ruskai2002analysis}, so we leave it for future work.

Utilising the counting argument of this section, we next obtain a characterisation of 4-way unital channels in terms of dual-unitary gates.


\subsection{Dual unitaries and 4-way unital quantum channels}

We start with the fact that \textit{unital} quantum channels can always be expressed as \textit{affine} combinations\footnote{Linear combinations with possibly negative coefficients summing to one.} of unitaries $U_k$~\cite{mendl2009unital}. That is, for completely positive maps:
\begin{subequations} \label{subeq:affine_unitary}
\begin{align}
 \text{(i) and (ii)} &\quad  \Longrightarrow  \\
\begin{split} \label{eq:unital_channel_mendl}
&\exists \quad \lambda_k \in \mathbb R, \; U_k \text{ unitary:}  \\
& q= \sum_k \lambda_k U_k \otimes U_k^*, \quad \sum_k \lambda_k=1 \,.
\end{split}
\end{align}
\end{subequations}
Conversely, any affine combination of unitaries yielding a completely positive map is clearly both trace-preserving and unital. 

On the other hand, restricting to \textit{convex} combinations of unitaries (i.e., $\lambda_k \ge 0$) defines \textit{mixed unitary} quantum channels~\cite{watrous2018theory}. As is well known, the sets of unital and mixed unitary quantum channels do not coincide, except for the case of a single qubit~\cite{landau1993birkhoff,mendl2009unital}. Although the convex representation is not exhaustive, it has the advantage of always producing valid (i.e., completely positive) quantum channels, a property that might fail in the affine form~\eqref{eq:unital_channel_mendl}.

Taking into account also space unitality, it is natural to define as \textit{mixed dual-unitary} the quantum channels that can be expressed as \textit{convex} combinations of dual-unitary gates. It is immediate from the definition that mixed dual-unitary channels are also 4-way unital. The same is also true for affine combinations of dual-unitary gates, but there complete positivity is not automatic, and needs to be demanded separately.

In fact, a converse of this last statement can also be established. Before stating the result, we need to recall an explicit parameterization of dual-unitary gates. In~\cite{claeys2021ergodic} it was shown that
\be  \label{eq:dual_unitary_param_main}
   U (J) = (W _1 \otimes W _2)  \exp \left(i J s_3 \otimes s_3\right) S (V _1 \otimes V _2),
\ee
are dual-unitary gates for any local dimension $d$, where $J\in \mathbb R$, the local unitaries are arbitrary, and $S$ denotes the swap operator. Above $s_i$ with $i  = 1,2,3$ denote the spin matrices for the $d$-dimensional irreducible representation of $su(2)$. Although the above paramaterization is exhaustive for $d=2$ dual-unitary gates~\cite{bertini2019exact}, it only represents a subfamily in higher dimensions.

We can now establish an implication analogous to~\eqref{subeq:affine_unitary} but for dual-unitary gates:


\begin{Proposition} \label{prop:qubits_affine_dual}
Any 4-way unital quantum channel can be decomposed as an affine combination of dual-unitary gates, i.e.,
\begin{align}
\label{eq:4way_dual_unitary}
 \text{(i)-(iv), } \;  &  \Longrightarrow  \nonumber \\
\begin{split}
&\exists \quad \lambda_k \in \mathbb R, \; U_k \text{ dual unitary:}  \\
& q= \sum_k \lambda_k U_k \otimes U_k^*, \quad \sum_k \lambda_k=1 \,.
\end{split}
\end{align}
In particular, dual-unitary gates of the form~\eqref{eq:dual_unitary_param_main} suffice.
\end{Proposition}

Our proof, given in Appendix~\ref{app:proofProp2}, is constructive and builds on top of the approach of Proposition~\ref{prop:counting_dimensions}. In a nutshell, the idea is to express all 4-way unitality preserving deformations around the completely depolarising channel as appropriate linear combinations of dual-unitary gates.
One of the difficulties is that an exhaustive parameterization for the latter beyond $d=2$ is not known. Nevertheless, the subfamily~\eqref{eq:dual_unitary_param_main} turns out to be sufficient. In particular, for $d>2$ there are dual-unitaries gates outside of the parameterization used in the proof. A unitary channel made out of such gates can be expressed as an affine combination of the before-mentioned subfamily, but not as a convex combination.

There are two direct corollaries of this Proposition. First, convex mixtures of dual-unitary gates are \textit{full-dimensional}\footnote{Meaning that they have the same dimension as the full set, but they are still non-exhaustive.} within the set of 4-way unital quantum channels. Moreover, by a direct application of Carath\'{e}odory's theorem~\cite{watrous2018theory}, we can bound the number of summands in the expansion of a mixed-dual unitary channel as convex combinations of dual-unitaries. The theorem ensures that a point in an $n$ dimensional convex hull can be expressed using at most $n+1$ points from that set.
In our case, using Table~\ref{tab:dimensions}, the theorem implies the existence of a positive integer
\begin{align}
    M = d^2 ( d^6 - 4 d^2 + 4)
\end{align}
such that every mixed dual-unitary channel can be written as
\begin{align}
    q= \sum_{k=1}^m p_k U_k \otimes U_k^*
\end{align}
where $m \le M$, $p$ is a probability vector and $U_k$ are dual-unitary.

One the other hand, it is not a priori clear whether every 4-way unital quantum channel is a \textit{convex} combination of dual-unitary gates. Physically, a positive answer would imply that all 4-way unital channels can be implemented through dual-unitary dynamics, together with classical randomness.


\subsection{Unbiased averaging around dual-unitary gates produces 4-way channels }

Our aim here is to provide one of the possible ways of obtaining a 4-way unital quantum channel, namely under averaging unitary dynamics. Remarkably, 4-way unitality can hold even if the unitary dynamics is non-dual-unitary. As we will see, the key idea is that it suffices that the underlying distribution is suitably symmetric around the dual-unitary family.

Specifically, we consider quantum channels of the form
\begin{align} \label{eq:convex_channel}
 \mathcal E_p = \int d\lambda \, p(\lambda)  U_
 \lambda(\cdot) U_\lambda^\dagger
\end{align}
where $\lambda$ is the collection of parameters specifying the unitary and $p(\lambda)$ is the associated probability density. Let us moreover assume the case of qubits $d=2$, as the so-called standard decomposition~\cite{kraus2001optimal} is crucial for our derivation. It states that every 2-qubit unitary can be explicitly parameterised as
\begin{align} \label{eq:standard_decomposition}
    U_\lambda &= W _1 \otimes W _2  \exp \left(i \sum_{j=1}^3 \theta_j  \sigma_j \otimes \sigma_j \right) V _1 \otimes V _2 ,
\end{align}
where the $\theta_j$ angles determine the nonlocal part and the local unitaries $V$, $W$ are arbitrary. As such, we denote $\lambda = (\lambda_W, \theta, \lambda_V)$ the parameters specifying the $W_1 \otimes W_2$, nonlocal, and $V_1 \otimes V_2$ parts of $U$ and $p(\lambda)$ the corresponding probability density.
\begin{enumerate}[]
    \item \textit{Assumption 1}: 
    \begin{align}
        p(\lambda) = p_{WV}(\lambda_W,\lambda_V) p_\theta(\theta_1,\theta_2,\theta_3)  \,,
    \end{align}
    i.e., the local and the nonlocal parts are independent.
\end{enumerate}
Now we proceed to impose unbiasedness of the distribution around dual-unitaries.  In the standard decomposition, the dual-unitary family corresponds to fixing $\theta_1=\theta_2 = \pi/4$.
\begin{enumerate}[]
    \item \textit{Assumption 2}: 
    \begin{align}
     p_\theta & (\pi/4+\delta_1,\pi/4 + \delta_2,\theta_3) \nonumber \\
    =  p_\theta &  (\pi/4-\delta_1,\pi/4+\delta_2,\theta_3) \nonumber \\ 
     =   p_\theta & (\pi/4+\delta_1,\pi/4-\delta_2,\theta_3) \,,
    \end{align}    
    i.e., the nonlocal part is unbiased around the dual-unitary family. 
\end{enumerate}

We are now ready to state our main result of this section.

\begin{Proposition} \label{prop:qubit_noise}
Under Assumptions~1 and 2, the quantum channel $\mathcal E_p$ is 4-way unital.
\end{Proposition}

The proof can be found in Appendix~\ref{app:averaging}.

Conceptually, Proposition~\ref{prop:qubit_noise} provides a simple model where suitable symmetry in the noise leads to solvability (as explored in Sec.~\ref{sec:applications}), \textit{even in the absence of dual-unitarity}.
This could provide an explanation why experimental realisations \cite{chertkov2021holographic, mi2021information} of dual-unitary dynamics are robust in the presence of unbiased errors. In particular, the errors do not destroy solvability and the particular features of zero correlations inside the light-cone are preserved. 

\subsection{Minimal example}

Next, we illustrate the introduced families and the effect of the different space-time unitality constraints on the simplest example, where the local folded space is two-dimensional. Starting with qubits, the folded space is spanned by the identity and three Pauli matrices. The minimal example arises from using single-qubit dephasing channels after each gate~\cite{kos2021correlations,kos2021chaos}. They project out two of the operators, for example $\sigma_1,\sigma_2$.
Explicitly, the projectors are given by
\be
P_{m}(\rho)= \sum_k^\infty e^{i \phi_{k} \sigma_3} \rho e^{-i \phi_{k} \sigma_3} \, ,
\ee
where $\phi_k$ are random i.i.d. phases. This correspond to an  averaged evolution intertwined by random space and time magnetic fields in the z-direction. 
We organise the basis as $\ket{\1\1}/2, \ket{\1\sigma_3}/2, \ket{\sigma_3\1}/2,\ket{\sigma_3\sigma_3}/2$ and write the resulting channel in this basis.

Since the condition (i) is always valid, it fixes the first row and the local gate can be written as
\begin{equation}
q = \begin{pmatrix}
    1 & 0 & 0 & 0\\
v_1 & \varepsilon_1 & a & b\\
v_2 & c & \varepsilon_2 & d\\
v_3 & e & f &  g
    \end{pmatrix}\!.
\label{eq:reducedgates}    
\end{equation}
Complete positivity is reflected in non-trivial constraints of $12$ real parameters with absolute size smaller than one. The additional conditions fix some of the parameters to zero. In particular, condition (ii) fixes $v_i=0$, (iii) $v_1=0, \varepsilon_1=0$ and (iv) $v_2=0, \varepsilon_2=0$. The additional conditions (v) and (vi) for a perfect channel fix $a=c=0$.

In the unital case (ii), we recover the example of averaged unitary evolution~\cite{kos2021correlations,kos2021chaos}, this time interpreted as a local quantum channel. Interestingly, complete positivity in this example is equivalent to the condition that the matrix is bistochastic after a Hadamard transformation~\cite{kos2021correlations}. This shows that this corresponds to a classical Markov chain problem.
One way this unital example arises is by taking all $\lambda_k$ from Eq.~\eqref{eq:unital_channel_mendl} to be identical and 
\be
U_k = (e^{i \phi_{k} \sigma_3} \otimes e^{i \phi_{k+1/2} \sigma_3})\cdot U \cdot (e^{i \phi_{2k} \sigma_3} \otimes e^{i \phi_{2k+1/2} \sigma_3}),
\label{eq:UbarSM}
\ee
with $U$ fixed and $\phi_{k}$ random. 
If $U$ is dual unitary, then $\varepsilon_1= \varepsilon_2=0$ and $q$ is 4-way unital.


\section{Applications}
\label{sec:applications}
Having introduced and characterised the new families of models, we move to answer questions about their dynamics. In particular, we use their additional constraints for simplifications of the calculation of the spatio-temporal correlation functions, spatial correlation functions after a quench, and steady states.

\subsection{Spatio-temporal correlation functions}

Consider spatio-temporal correlations between single-site traceless Hilbert-Schmidt normalised operators\footnote{Normalised as $\Tr (a a^\dagger) =1$.} $a$ and $b$ starting from an infinite temperature steady state $\rho_{\infty}= \1/d^{2L}$. In the case of unital quantum channels, this is the general time translation invariant state. The correlations thus read
\begin{align}
\braket{a_i(t) b_j} &\equiv	d \ \Tr \left( [a_i(t)]^\dagger b_j \frac{\1}{d^{2L}} \right)\notag \\
&= \bra{\mcirc \ldots a \ldots \mcirc} \mathbb{Q}^t \ket{\mcirc \ldots b \ldots \mcirc},
\end{align}
where we included the factor of $d$ in the definition such that the auto-correlation function at time $0$ is normalised to one.

The following is a simple generalisation of the exact calculation of the correlation functions of dual-unitary evolution~\cite{bertini2019exact}. It was already pointed out in~\cite{kos2021correlations} that the simplifying property is dual-unitality, which in those cases followed from the dual-unitarity\footnote{A simple way to see this is by taking the unitary gate as the only Kraus operator or by looking at the graphical rules in the folded picture.}. In our setting, in contrast to the works mentioned previously, dual-unitarity might actually be absent, and moreover unitality conditions can hold in less than four directions.

The correlations in the folded picture are graphically expressed as
\begin{align}
&\braket{a_i(t) b_j}=
\begin{tikzpicture}[baseline=(current  bounding  box.center), scale=0.55]
\foreach \jj[evaluate=\jj as \j using -2*(ceil(\jj/2)-\jj/2)] in {0,...,3}
\foreach \i in {1,...,4}
{\Wgategreen{2*\i+\j}{\jj}}
\foreach \i in {2,...,9}{
\MYcircle{\i-.5}{-0.5}
\MYcircle{\i-1.5}{4-0.5}
}
\MYcircleB{3.5}{-.5}
\MYcircleB{4.5}{4-.5}
\Text[x=3.5,y=-1.0]{$b$}
\Text[x=4.5,y=4.0]{$a$}
\end{tikzpicture}\, ,
\label{eq:Corr1}
\end{align}
where we consider $L>2t$, and the black dots correspond to the operators $a$ and $b$.
First we repeatedly apply condition (i), resulting in a backwards light-cone
\begin{align}
&\braket{a_i(t) b_j} =
\begin{tikzpicture}[baseline=(current  bounding  box.center), scale=0.55]
\foreach \i in {1,...,4}{
\Wgategreen{2*\i}{0}}
\foreach \i in {1,...,3}{
\Wgategreen{2*\i+1}{1}}
\foreach \i in {2,...,3}{
\Wgategreen{2*\i}{2}}
\Wgategreen{2*2+1}{3}
\foreach \i in {2,...,9}{
\MYcircle{\i-.5}{-0.5}}
\foreach \i in {2,...,5}{
\MYcircle{\i-.5}{-0.5+\i-1}
\MYcircle{\i-.5+4}{-0.5-\i+6}}
\MYcircleB{3.5}{-.5}
\MYcircleB{4.5}{4-.5}
\Text[x=3.5,y=-1.0]{$b$}
\Text[x=4.5,y=4.0]{$a$}
\end{tikzpicture}\, .
\label{eq:Corr2}
\end{align}
If $|i-j| > t$, the operator $b$ lies outside the light-cone and is contracted with $\ket{\mcirc}$. Since $b$ is traceless, the correlations are exactly zero in this case.

Assuming the left unitality condition (iii), the correlations simplify further: 
\begin{align} 
&\braket{a_i(t) b_j} =
\begin{tikzpicture}[baseline=(current  bounding  box.center), scale=0.55]
\foreach \i in {2,...,4}{
\Wgategreen{2*\i}{0}}
\foreach \i in {2,...,3}{
\Wgategreen{2*\i+1}{1}}
\foreach \i in {3,...,3}{
\Wgategreen{2*\i}{2}}
\Wgategreen{2*2+1}{3}
\foreach \i in {4,...,9}{
\MYcircle{\i-.5}{-0.5}
}
\foreach \i in {2,...,3}{
\MYcircle{\i-.5+2}{-0.5+\i-1}
}
\foreach \i in {2,...,5}{
\MYcircle{\i-.5+4}{-0.5-\i+6}
}
\MYcircleB{3.5}{-.5}
\MYcircleB{4.5}{4-.5}
\MYcircle{4.5}{2.5}
\Text[x=3.5,y=-1.0]{$b$}
\Text[x=4.5,y=4.0]{$a$}
\end{tikzpicture}\, .
\label{eq:Corr3}
\end{align}
Notice that if $i$ would be an integer or $j$ half-integer and $|i-j| \neq t$, the correlations would be zero. For example, in the case of $b$ at the left bottom half-integer position, we would perform the simplification
\be
\begin{tikzpicture}[baseline=(current  bounding  box.center), scale=.7]
\Wgategreen{0}{0}
\MYcircle{-.5}{-.5}
\MYcircle{-.5}{0.5}
\MYcircleB{.5}{-.5}
\Text[x=1.00,y=-0.50]{$b$}
\end{tikzpicture}\,  =
\begin{tikzpicture}[baseline=(current  bounding  box.center), scale=.7]
\draw[thick](-.5,-.5)--(.5,-.5);\draw[thick](-.5,.5)--(.5,.5);
\MYcircle{-.5}{.5}
\MYcircle{-.5}{-.5}
\MYcircleB{.5}{-.5}
\Text[x=1.0,y=-.5]{$b$}
\end{tikzpicture}=0,
\ee
where the last equality follows from $\Tr b=0$.

Notice that conditions (i) and (iii) are not always enough to simplify Eq.~\eqref{eq:Corr3} to an efficiently computable expression. 
In particular, the expression may still contain proportional to $t$ many gates in the bottom layer. Thus an exact evaluation may be exponentially hard in $t$.

\subsubsection{3-way space unitality}
Assuming the right unitality condition (iv), we simplify the diagram also from the right side. Therefore we obtain an analogue of the dual-unitary result~\cite{bertini2019exact}, even without assuming (time) unitality. The correlations are exactly equal to zero, except if they lie on the same light ray.
An example of non-zero correlation is:
\begin{align}
&\braket{a_i(t) b_j} =
\begin{tikzpicture}[baseline=(current  bounding  box.center), scale=0.55]
\Wgategreen{2*2}{0}\Wgategreen{2*2+1}{1}
\Wgategreen{2*3}{2}\Wgategreen{7}{3}
\foreach \i in {2,...,5}{
\MYcircle{\i-.5+2}{-0.5+\i-1}
\MYcircle{\i-.5+4-1}{-0.5+\i-2}
}
\MYcircleB{3.5}{-.5}
\MYcircleB{7.5}{4-.5}
\Text[x=3.5,y=-1.0]{$b$}
\Text[x=7.5,y=4.0]{$a$}
\end{tikzpicture}\, .
\label{eq:Corr40}
\end{align}
These correlations are expressed using a single qudit channels
\begin{align}
&\mathcal{M}_+ =
\begin{tikzpicture}[baseline=(current  bounding  box.center), scale=0.55]
\Wgategreen{0}{0}
\MYcircle{.5}{-.5}
\MYcircle{-.5}{.5}
\end{tikzpicture}\, , \quad
&\mathcal{M}_- =
\begin{tikzpicture}[baseline=(current  bounding  box.center), scale=0.55]
\Wgategreen{0}{0}
\MYcircle{.5}{.5}
\MYcircle{-.5}{-.5}
\end{tikzpicture}\, ,
\label{eq:MapsM}
\end{align}
which can be expressed as (non-vectorised) superoperators
\begin{subequations}
\begin{align}
    \mathcal{M}_+ (o) &= \Tr_1 \left(  \sum_k E_k (o \otimes \1) E_k^\dagger \right) 
    \, , \\
    \mathcal{M}_- (o) &= \Tr_2 \left(  \sum_k E_k ( \1 \otimes o) E_k^\dagger \right) 
    \, .
\end{align}
\end{subequations}
%
The correlations are thus equal to:
\begin{align}\label{eq:corr0}
&\braket{a_i(t) b_j}_{\mathrm{lr}} = \\
&\delta_{|i-j|,t}\left(\delta_{j \in \mathbb{Z}} \bra{a} \mathcal{M}_+^{2t} \ket{b} + \delta_{(\frac{1}{2}+j) \in \mathbb{Z}} \bra{a} \mathcal{M}_-^{2t} \ket{b} \right) . \notag
\end{align}
We introduced the subscript $\mathrm{lr}$ (light ray) for later reference.
Since this expression only involves actions of single qudit channels, it can be evaluated efficiently. Notice that even though $q$ may be non-unital, $\mathcal{M}_\pm$ are always both trace-preserving and unital, as a consequence of the space unitality conditions.

We see that 3-way space unitality and 4-way unitality show the same physical behaviour of the correlation functions as in dual-unitary circuits, even-though the former cases are much more general. In particular, correlations stay exactly zero inside the light cone. Moreover, the correlations on the light ray pinpoint  the fact that the Lieb-Robinson velocity of these circuits is the maximal one allowed by the geometry \cite{claeys2020maximum}. 

The long time behaviour of the correlations will be given by the leading eigenvalue of $\mathcal{M_\pm}$. Generically, the correlations on the light ray will decay exponentially, signalling ergodic and mixing behaviour, but one can recover the ergodic hierarchy, as in the dual-unitary case~\cite{bertini2019exact,claeys2021ergodic, aravinda2021from}. 
In contrast to the generic behaviour, all or some correlations might remain constant in time corresponding to noninteracting and nonergodic behaviour, respectively. Or some correlations may remain oscillating, corresponding to ergodic but nonmixing behaviour.
Here we can also observe the so-called Bernoulli case even for qubits, where all correlations decay to zero in one time step, in contrast to dual-unitaries for qubits \cite{aravinda2021from}.

\subsubsection{3-way time unitality}
\label{sec:3wayunitality}
Instead, if we assume conditions (i)-(iii), the situation again simplifies, but with a slightly richer behaviour of correlations. In particular, the configuration of Eq.~\eqref{eq:Corr3} simplifies from the bottom up and gives non-zero correlations
\begin{align}
&\braket{a_i(t) b_j} =
\begin{tikzpicture}[baseline=(current  bounding  box.center), scale=0.55]
\Wgategreen{2*2}{0}
\Wgategreen{2*2+1}{1}
\Wgategreen{2*3}{2}
\Wgategreen{2*2+1}{3}
\foreach \i in {4,...,5}{
\MYcircle{\i-.5}{-0.5}
}
\foreach \i in {2,...,3}{
\MYcircle{\i-.5+2}{-0.5+\i-1}
\MYcircle{\i-.5+4}{-0.5+\i-1}
}
\MYcircleB{3.5}{-.5}
\MYcircleB{4.5}{4-.5}
\MYcircle{4.5}{2.5}
\MYcircle{5.5}{3.5}
\MYcircle{6.5}{2.5}
\Text[x=3.5,y=-1.0]{$b$}
\Text[x=4.5,y=4.0]{$a$}
\end{tikzpicture}\, .
\label{eq:Corr4}
\end{align}
This expression generalises the particular solvable family of half-dual-unitary reduced gates from~\cite{kos2021correlations}.
The expression contains $\mathcal{E}_2$, one of the two additional single qudit channels 
\begin{align}
&\mathcal{E}_1 =
\begin{tikzpicture}[baseline=(current  bounding  box.center), scale=0.55]
\Wgategreen{0}{0}
\MYcircle{-.5}{-.5}
\MYcircle{-.5}{.5}
\end{tikzpicture}\, , \quad
&\mathcal{E}_2 =
\begin{tikzpicture}[baseline=(current  bounding  box.center), scale=0.55]
\Wgategreen{0}{0}
\MYcircle{.5}{.5}
\MYcircle{.5}{-.5}
\end{tikzpicture}\, .
\label{eq:MapsM2}
\end{align}
Explicitly, $\mathcal{E}_2$ as a (non-vectorised) superoperator reads
\begin{align}
    \mathcal{E}_2 (o)= \Tr_2 \left( \sum_k E_k (o \otimes \1) E_k^\dagger \right) 
    \, .
\end{align} 
Since here we assumed (iii), $\mathcal{E}_1$ is trivial. However, its role becomes important if one assumes (iv) instead of (iii).
The complete expression for correlations thus reads (for $|i-j|<t$)
\begin{align}
&\braket{a_i(t) b_j}_n =\\
&\delta_{j \in \mathbb{Z}} \delta_{(\frac{1}{2}+i)\in \mathbb{Z}}
\bra{a} \mathcal{M}_-^{t-(i-j)-\frac{1}{2}}\mathcal{E}_2\mathcal{M}_+^{t+i-j-\frac{1}{2}}\! \ket{b}
, \notag
\end{align}
where we defined $\braket{a_i(t) b_j}_n = \braket{a_i(t) b_j} - \braket{a_i(t) b_j}_{\mathrm{lr}}$ [cf. \eqref{eq:corr0}]. This can be easily evaluated using only single qudit channels. 
If instead we assume conditions (i), (ii), and (iv), we get similar diagrams, but this time the evolution starts by going left with a possible turn to the right. It reads (for $|i-j|<t$)
\begin{align}
&\braket{a_i(t) b_j}_n =\\
&\delta_{(\frac{1}{2}+j)\in \mathbb{Z}}\delta_{i \in \mathbb{Z}}  \bra{a} \mathcal{M}_+^{t+i-j-\frac{1}{2}} \mathcal{E}_1\mathcal{M}_-^{t-(i-j)-\frac{1}{2}} \ket{b} . \notag
\end{align}

As we have seen, for the time-unital local channels it is enough to demand only one of the space unitality conditions to simplify the expressions. Indeed, the resulting equations contain only single qudit channels, which can be efficiently evaluated. The resulting correlations are, in contrast to 3-way space unitality and dual-unitarity, non-zero inside the light cone between operators on different sub-lattices (integers or half integers). Nevertheless, there are still non-zero correlations along the light ray, suggesting the maximal Lieb-Robinson velocity.

More concretely, let us assume that $\mathcal{M}_\pm$ have $\ket{\circ}$ as their unique leading eigenvector (corresponding to eigenvalue one) and no other eigenvalue of the same modulus. Denote as $\lambda_\pm$ the sub-leading eigenvalues. Then the correlations, when they are non-zero, behave as
\begin{align}
&\braket{a_i(t) b_j}_n \sim (\lambda_+ \lambda_-)^t \left( \frac{\lambda_+}{\lambda_-} \right)^{i-j} \, .
\end{align}
Therefore at fixed $t$ the correlations grow or decay exponentially in $i-j$, depending on the ratio of the eigenvalues. This contrasts the zero correlations inside the light-cone that appear in the case of 4-way unital channels and dual-unitaries.

\subsection{Spatial correlation functions after a quench} \label{sec:corr_fun_quench}
Next, we focus on spatial correlations between single-site operators after a quantum quench starting from an initial density matrix $\rho(0)$:
\begin{align}
C_{ij}(t) \  \equiv \
&  \ d \ \Tr a_i b_j \rho(t) 
=\\
&\bra{\mcirc \ldots a\mcirc \ldots b \ldots \mcirc}\mathbb{Q}^t \ket{\rho(0)}. \notag
\end{align}
Without loss of generality we assume $i<j$.

For the purpose of presentation, let us restrict ourselves to a two sites translation-invariant initial states. The discussion straightforwardly extends to the non-transnational invariant case. 
Without loss of generality we write them as a matrix product density operator (MPDO) in a vectorised form
\begin{align}
\ket{\rho(0)} &= \!\! \sum_{s_{\frac{1}{2}}, \ldots,s_L}  \!\! \Tr [T^{(s_{\frac{1}{2}},s_1)} \ldots T^{(s_{L-\frac{1}{2}},s_L)}] \ket{s_{\frac{1}{2}} \ldots s_L} \nonumber \\
&=
\frac{1}{d^{L}} \;
\begin{tikzpicture}[baseline=(current  bounding  box.center), scale=.65]
\draw[very thick] (-0.5,0) -- (4.5,0);
\draw[very thick] (-0.5,0.15) arc (90:270:0.075);
\draw [very thick, dotted] (4.5,.0) -- (5.5,.0);
\draw[very thick] (5.5,0.15) arc (90:-90:0.075);
\foreach \i in {0,1.35,2.7,4.05}
{
\rhoO{\i}{0}}
\end{tikzpicture}\,.
\label{eq:quench1}
\end{align}
The state is defined through $\chi \times \chi$ matrices $T^{(s_{i-1/2},s_i)}$ (green boxes).
In the diagram the horizontal line represents the \textit{auxiliary} space of bond dimension $\chi$. 
The MPDO is taken to be normalized in the thermodynamic limit
\be \label{eq:mpdo_normalisation}
\lim_{L\to \infty} \Tr \rho(0)  = \lim_{L\to \infty} \Tr E(0)^L=1 ,
\ee
meaning that the space transfer matrix 
 \be \label{eq:transfer_matrix_def}
 E(0)= \,\,\begin{tikzpicture}[baseline=(current  bounding  box.center), scale=.75]
\rhoO{0}{0}
\MYcircle{-0.5}{.5}
\MYcircle{.5}{.5}
\end{tikzpicture}
 \ee
has a unique fixed point with eigenvalue one, which we denote by $\ket{\triangle}\bra{\square}$, i.e.,
\begin{align} \label{eq:solvable_mpdo_fixed_point}
    \lim_{L\to \infty} E(0)^L= \ket{\triangle}\bra{\msquare} .
\end{align}
with $\braket{\square|\triangle}=1$.
%

Graphically, we express the correlations after time $t$ as: 
\begin{align}\label{eq:quench2}
&C_{ij}(t) = \begin{tikzpicture}[baseline=(current  bounding  box.center), scale=0.6]
\foreach \i in {0,...,1}{
\Wgategreen{-2}{0+2*\i}\Wgategreen{0}{0+2*\i}\Wgategreen{2}{0+2*\i}\Wgategreen{4}{0+2*\i}\Wgategreen{6}{0+2*\i}
\Wgategreen{-1}{1+2*\i}\Wgategreen{1}{1+2*\i}\Wgategreen{3}{1+2*\i}\Wgategreen{5}{1+2*\i}\Wgategreen{7}{1+2*\i}
}
\draw[very thick] (-2.25,-1) -- (7.5,-1);
\foreach \i in {2.5,4.5,6.5,8.5,10.5}
{\rhoO{-3.5+\i}{-1}}
 \foreach \i in {-1,...,8}{
 \draw[thick, fill=white] (\i-0.5,4-0.5) circle (0.1cm); 
}
\MYcircleB{7.5}{3.5}
\MYcircleB{-1.5}{3.5}
\Text[x=-1.5,y=4.0]{$a$}
\Text[x=7.5,y=4.0]{$b$}
\end{tikzpicture}.
\end{align}
Next, we use condition (i), which results in two backward propagating light cones originating at the positions of the operators, connected by two parts of the circuit consisting of powers of the transfer matrix $E(0)$:
\bw
\begin{align}
C_{ij}(t)=
\begin{tikzpicture}[baseline=(current  bounding  box.center), scale=0.6]
\Wgategreen{-4}{0}\Wgategreen{-2}{0}
\Wgategreen{0}{0}\Wgategreen{2}{0}\Wgategreen{4}{0}\Wgategreen{6}{0}\Wgategreen{8}{0}\Wgategreen{10}{0}
\Wgategreen{-1}{3}\Wgategreen{7}{3}
\Wgategreen{-2}{2}\Wgategreen{0}{2}\Wgategreen{6}{2}\Wgategreen{8}{2}
\Wgategreen{-3}{1}\Wgategreen{-1}{1}\Wgategreen{1}{1}\Wgategreen{5}{1}\Wgategreen{7}{1}\Wgategreen{9}{1}
\draw[very thick] (-6.5,-1) -- (13.5,-1);
\draw[very thick,dotted] (-8.5,-1) -- (-6.5,-1);
\draw[very thick,dotted] (13.5,-1) -- (14.5,-1);
\foreach \i in {-3.5,-1.5,0.5,2.5,4.5,6.5,8.5,10.5,12.5,14.5,16.5}
{
\rhoO{-3.5+\i}{-1}
}
 \foreach \i in {0,...,3}{
\MYcircle{\i-.5}{3.5-\i}
\MYcircle{\i+4-.5}{0.5+\i}
\MYcircle{\i-.5+8}{3.5-\i}
\MYcircle{\i-4-.5}{0.5+\i}
}
 \MYcircle{-5.5}{-.5}
  \MYcircle{-6.5}{-.5}
   \MYcircle{-7.5}{-.5}
 \MYcircle{11.5}{-.5}
 \MYcircle{12.5}{-.5}
 \MYcircle{13.5}{-.5}
\MYcircleB{7.5}{3.5}
\MYcircleB{-1.5}{3.5}\Text[x=-1.5,y=4.0]{$a$}
\Text[x=7.5,y=4.0]{$b$}
\end{tikzpicture}.
\label{eq:quench4}
\end{align}
\ew
In our particular example, we have only one $E(0)$ in the middle, but in general the number of $E(0)$ depends on the separation\footnote{Remember that $j$ is an integer, $i$ a half-integer, and $i<j$.} $j-i$ as $j-i-2t+\frac{1}{2}$.
In the thermodynamic limit we can substitute the outer part of the MPDO using Eq.\eqref{eq:solvable_mpdo_fixed_point}.

The above discussion assumed periodic boundary conditions, where we needed to take $L$ to be large. If instead we consider open boundary conditions with completely depolarising channels at the edges, we can perform simplifications and obtain the exact same expression even for finite systems. We need also to specify the boundary tensors for MPDO, which need to be $\msquare$ and $\triangle$ at the left and right boundary respectively.

We have already vastly simplified the expressions, but the evaluation of the left/right light cone can in principle still be exponentially (in $t$) hard.
To streamline the expression further we need to contract the left and right edges,  
\begin{tikzpicture}[baseline=-1.5pt, scale=.5]
\rhoO{0}{0}
\MYsquare{-0.5}{-.0}\MYcircle{-.5}{.5}
\end{tikzpicture}
and 
\begin{tikzpicture}[baseline=-1.5pt, scale=.5]
\rhoO{0}{0}
\MYtriangle{0.5}{-.0}\MYcircle{.5}{.5}
\end{tikzpicture}.
This leads us to define that a MPDO is \emph{solvable} if the following equations hold 
\be
 \begin{tikzpicture}[baseline=(current  bounding  box.center), scale=.75]
\rhoO{0}{0}
\MYsquare{-0.5}{-.0}\MYcircle{-.5}{.5}
\end{tikzpicture}= \begin{tikzpicture}[baseline=(current  bounding  box.center), scale=.75]
\draw[very thick] (-0.5,0.5)--(-.3,0.7);
\draw[very thick] (-0.5,0)--(-.3,0);
\MYsquare{-0.5}{-.0}\MYcircle{-.5}{.5}
\end{tikzpicture}, 
\qquad
 \begin{tikzpicture}[baseline=(current  bounding  box.center), scale=.75]
\rhoO{0}{0}
\MYtriangle{0.5}{-.0}\MYcircle{.5}{.5}
\end{tikzpicture}= \begin{tikzpicture}[baseline=(current  bounding  box.center), scale=.75]
\draw[very thick] (0.5,0.5)--(.3,0.7);
\draw[very thick] (0.5,0)--(.3,0);
\MYtriangle{0.5}{-.0}\MYcircle{.5}{.5}
\end{tikzpicture}
\label{eq:cond0State}\; .
 \ee

In this way, we arrived at a generalisation of the initial solvable matrix product states (MPS) introduced in~\cite{piroli2020exact} for pure states and dual-unitary evolution.
In that case, the solvable MPS were classified by the connection that their local tensors are unitaries acting on the combined auxiliary and physical space. 
We make an analogue classification for MPDO with a local purification in the next section~\ref{sec:MPDO}. In this section we take the conditions as given, and derive exact results using them.

If our local channels satisfy conditions (iii) or (iv), we can use them together with the MPDO solvability condition \eqref{eq:cond0State} to simplify the left or right light cone correspondingly.
When simplifying left or right light cone, we obtain zero if $i$ is integer or $j$ half-integer, respectively.
Assuming both (iii) and (iv), we simplify both parts and obtain for our particular positions of $a$ and $b$
\be
C_{ij}(t)=
\begin{tikzpicture}[baseline=(current  bounding  box.center), scale=0.6]
\Wgategreen{2}{0}\Wgategreen{4}{0}
\Wgategreen{-1}{3}\Wgategreen{7}{3}
\Wgategreen{0}{2}\Wgategreen{6}{2}
\Wgategreen{1}{1}\Wgategreen{5}{1}
\rhoO{3}{-1}
\foreach \i in {0,...,3}{
\MYcircle{\i-.5}{3.5-\i}
\MYcircle{\i-.5-1}{3.5-\i-1}
\MYcircle{\i+4-.5}{0.5+\i}
\MYcircle{\i+4-.5+1}{0.5+\i-1}
}
\MYtriangle{3.5}{-1}
 \MYsquare{2.5}{-1}
\MYcircleB{7.5}{3.5}
\MYcircleB{-1.5}{3.5}\Text[x=-1.5,y=4.0]{$a$}
\Text[x=7.5,y=4.0]{$b$}
\end{tikzpicture}.
\label{eq:quench5}
\ee
Instead, if the separation $j-i$ is shorter than $2t$, we obtain
\be
C_{ij}(t)=
\begin{tikzpicture}[baseline=(current  bounding  box.center), scale=0.6]
\Wgategreen{-1}{3}\Wgategreen{5}{3}
\Wgategreen{0}{2}\Wgategreen{4}{2}
\Wgategreen{1}{1}\Wgategreen{3}{1}
\Wgategreen{2}{0}
\foreach \i in {0,...,2}{
\MYcircle{\i-.5}{3.5-\i}
\MYcircle{\i-.5-1}{3.5-\i-1}
\MYcircle{\i+4-.5-1}{0.5+\i+1}
\MYcircle{\i+4-.5+1-1}{0.5+\i}
}
\MYcircle{1.5}{-0.5}
\MYcircle{2.5}{-0.5}
\MYcircleB{5.5}{3.5}
\MYcircleB{-1.5}{3.5}\Text[x=-1.5,y=4.0]{$a$}
\Text[x=5.5,y=4.0]{$b$}
\end{tikzpicture}.
\label{eq:quench6}
\ee
Both kinds of expressions can be efficiently evaluated as actions of single qudit channels $\mathcal{M}_{\pm}$. For $i$ half-integer, $j$ integer and separation $s=j-i$, the result reads:
\be
\! C_{ij}(t)\!=\!
\begin{cases}
\bra{ab}( \mathcal{M}_- \otimes \mathcal{M}_+)^{2t}\ket{I(s-2t+\frac{1}{2})}\!, s > 2t,
\\
\bra{ab}( \mathcal{M}_- \otimes \mathcal{M}_+)^{s-\frac{1}{2}}q\ket{\mcirc\mcirc}, \ \ s < 2t,
\end{cases}
\ee
where 
\be
\ket{I(x)}=
\begin{cases}
 \begin{tikzpicture}[baseline=(current  bounding  box.center), scale=0.6]
\rhoO{3}{-1}
\MYtriangle{3.5}{-1}
 \MYsquare{2.5}{-1}
\end{tikzpicture},
\ \ x=1, \\
 \begin{tikzpicture}[baseline=(current  bounding  box.center), scale=0.6]
  \draw[very thick, dotted] (0,-1)-- (2,-1);
  \draw[very thick, dotted] (4,-1)-- (6,-1);
 \draw[very thick] (2,-1)-- (4,-1);
\rhoO{0}{-1} \rhoO{2}{-1} \rhoO{4}{-1} \rhoO{6}{-1}
\MYtriangle{6.5}{-1} \MYsquare{-.5}{-1}
 \MYcircle{0.5}{-.5}  \MYcircle{1.5}{-.5}
  \MYcircle{2.5}{-.5}   \MYcircle{3.5}{-.5}
  \MYcircle{4.5}{-.5}  \MYcircle{5.5}{-.5}
   \draw [thick, black,decorate,decoration={brace,amplitude=2.5pt,mirror},xshift=0.0pt,yshift=-0.0pt](1.4,-1.2) -- (4.6,-1.2) node[black,midway,yshift=-0.35cm] {$x-2$ times $E(0)$};
   \Text[x=1.5,y=.3]{}
\end{tikzpicture},
\ \ x>1.
\end{cases}
\ee

Notice that for $j-i<2t$ there is no information about the initial state (cf. \eqref{eq:quench6}). This means that we obtained exact expressions for the two-point correlation functions in the steady state, which is reached from all solvable initial states. We elaborate on this point in the following subsection.
Since we did not use condition (ii) anywhere, the whole evolution might be non-unital and these correlations could be non-trivial. This is different from the dual-unitary case, where the evolution is unital and correlations are exactly zero for $j-i<2t$~\cite{piroli2020exact}. If, as at the end of section \ref{sec:3wayunitality}, we assume that $\lambda_\pm$ are the unique  sub-leading eigenvalues of $\mathcal{M}_\pm$, the correlations decay exponentially in the separation as:
\begin{align}
    C_{ij}(t) \sim (\lambda_+ \lambda_-)^{j-i}\, .
\end{align}
We can understand this as correlations which are created in the past and then propagated ballistically with exponential damping.  

Let us also comment a bit more on the case $j-i\geq2t$. Here, if we additionally assume $\lambda_0$ to be the unique sub-leading eigenvalue of $E(0)$, the correlations decay exponentially as:
\begin{align}
   C_{ij}(t) \sim (\lambda_+ \lambda_-)^{2t} \lambda_0^{j-i-2t} \, .
\end{align}
We can see that the exponential decay in the additional separation $j-i-2t$ originates from the exponential decay in the initial MPDO. If the initial MPDO is instead uncorrelated at such distance, the correlations after the time evolution also vanish directly due to locality.

\subsection{Non-equilibrium steady state} \label{sec:noneq_steady_st}

Here we focus on the evolution of solvable MPDO at long times, when the dynamics obeys different unitality conditions. In particular, we provide exact expressions for the steady state and examine its uniqueness, provided that we start from solvable states. Remarkably, we find that steady states can be reached exactly in time of order of the system size. This does not necessarily imply that the steady state is unique if we start from general initial states. 
For instance, for an integrable model such as the dual-unitary transverse-field Ising model, the steady state from solvable initial states is unique (infinite temperature state). But the steady state is not unique in general, as it is given by the Generalised Gibbs Ensemble~\cite{vidmar2016generalized} due to the presence of conserved charges.

We show uniqueness of the steady state from solvable MPDOs in two cases:
\begin{itemize}
        \item[(1.)] Open boundary conditions with completely depolarising channels $\ket{\mcirc}\bra{\mcirc}$ at the edges. Here taking just one of the left or right unitality conditions suffices, i.e., 2-way unitality. More concretely, the evolution is given by:
    \begin{align}
        \mathbb{Q} = (\ket{\mcirc}\bra{\mcirc} \otimes q^{\otimes L})(  q^{\otimes L}\otimes \ket{\mcirc}\bra{\mcirc}).
    \end{align}
   Also, the left and right boundary tensors of the MPDO need to be contracted with $\msquare$ and $\triangle$.
    \item[(2.)] We focus on a finite region $A$ in the limit of infinite system size $L \to \infty$. We need both the left and right unitality conditions here, i.e., we need to assume 3-way space unitality.
\end{itemize}
Moreover, we show that the steady state is reached exactly in a finite time $t\leq|A|$, where $|A|$ is the (sub)system length ($A$ contains $2|A|$ qudits).

Focusing first on (1.), the evolution is given by 
\begin{align}
\ket{\rho(t)} \!=\!
\begin{tikzpicture}[baseline=(current  bounding  box.center), scale=0.6]
\foreach \i in {0,...,3}{
\Wgategreen{0}{0+2*\i}\Wgategreen{2}{0+2*\i}\Wgategreen{4}{0+2*\i}\Wgategreen{6}{0+2*\i}
\Wgategreen{-1}{1+2*\i}\Wgategreen{1}{1+2*\i}\Wgategreen{3}{1+2*\i}\Wgategreen{5}{1+2*\i}\Wgategreen{7}{1+2*\i}
}
  \foreach \i in {-0.5,...,6.5}{
\draw[very thick] (\i, 7.5) -- (\i,7.75);
}
\foreach \j in {-0.5,...,7.5}{
 \MYcircle{-1.5}{\j}
 \MYcircle{7.5}{\j}
}
\draw [loosely dotted] (-2,-.75) -- (8,-.75);
\end{tikzpicture},
\label{eq:NESS1}
\end{align}
where the dots denote the omitted part of the evolution. Next, we use left unitality (iii), the first of the conditions~\eqref{eq:cond0State} and that $t\geq|A|$. The expression simplifies to
\begin{align}
\ket{\rho(t)} =
\begin{tikzpicture}[baseline=(current  bounding  box.center), scale=0.6]
\foreach \i in {3}{
\Wgategreen{2}{0+2*\i}\Wgategreen{4}{0+2*\i}\Wgategreen{6}{0+2*\i}
\Wgategreen{1}{1+2*\i}\Wgategreen{3}{1+2*\i}\Wgategreen{5}{1+2*\i}\Wgategreen{7}{1+2*\i}
}
\foreach \i in {2}{
\Wgategreen{4}{0+2*\i}\Wgategreen{6}{0+2*\i}
\Wgategreen{3}{1+2*\i}\Wgategreen{5}{1+2*\i}\Wgategreen{7}{1+2*\i}
}
\foreach \i in {1}{
\Wgategreen{6}{0+2*\i}
\Wgategreen{5}{1+2*\i}\Wgategreen{7}{1+2*\i}
}
\Wgategreen{7}{1}
  \foreach \i in {-0.5,...,6.5}{
\draw[very thick] (\i, 7.5) -- (\i,7.75);
}
\foreach \j in {0.5,...,7.5}{
  \MYcircle{7.5}{\j}
  \MYcircle{7-\j}{\j}
}
\end{tikzpicture}.
\label{eq:NESS2}
\end{align}
The diagram fully characterises the steady state.
The lack of the initial state in the above expression implies that steady state is unique. 
In the case of (time) unitality the diagram simplifies from the bottom up, and the unique steady state is the infinite temperature state $\ket{\mcirc \ldots \mcirc}$.

If, in addition, the right unitality condition (iv) holds, the expression simplifies from the right
\begin{align}
\ket{\rho(t)} =
\begin{tikzpicture}[baseline=(current  bounding  box.center), scale=0.6]
\foreach \i in {3}{
\Wgategreen{2}{0+2*\i}\Wgategreen{4}{0+2*\i}
\Wgategreen{1}{1+2*\i}\Wgategreen{3}{1+2*\i}\Wgategreen{5}{1+2*\i}
}
\Wgategreen{3}{5}
 \foreach \i in {-0.5,...,6.5}{
\draw[very thick] (\i, 7.5) -- (\i,7.75);
}
\foreach \j in {4.5,...,7.5}{
  \MYcircle{-1+\j}{\j}
  \MYcircle{7-\j}{\j}
}
\end{tikzpicture}.
\label{eq:NESS3}
\end{align}

The same expression is obtained in (2.) in the following way.
We focus on the reduced density matrix of the sub-system region $A$:
\begin{align}
\rho_A(t)= \Tr_{\Bar{A}} \rho(t),
\end{align}
where the bar denotes the complement of region $A$.
First we use trace preservation (i) to simplify the region outside of $A$.
We obtain half light-cones starting outwards of the region of interest, with $\mcirc$ operators on it:
\bw
\begin{align}
\ket{\rho_A(t)}= d^{|A|}
\begin{tikzpicture}[baseline=(current  bounding  box.center), scale=0.6]
\Wgategreen{-4}{0}\Wgategreen{-2}{0}
\Wgategreen{0}{0}\Wgategreen{2}{0}\Wgategreen{4}{0}\Wgategreen{6}{0}\Wgategreen{8}{0}\Wgategreen{10}{0}
\Wgategreen{-1}{3}\Wgategreen{1}{3}\Wgategreen{3}{3}\Wgategreen{5}{3}\Wgategreen{7}{3}
\Wgategreen{-2}{2}\Wgategreen{0}{2}\Wgategreen{2}{2}\Wgategreen{4}{2}\Wgategreen{6}{2}\Wgategreen{8}{2}
\Wgategreen{-3}{1}\Wgategreen{-1}{1}\Wgategreen{1}{1}\Wgategreen{3}{1}\Wgategreen{5}{1}\Wgategreen{7}{1}\Wgategreen{9}{1}
\draw[very thick] (-6.5,-1) -- (13.5,-1);
\draw[very thick,dotted] (-8.5,-1) -- (-6.5,-1);
\draw[very thick,dotted] (13.5,-1) -- (14.5,-1);
\foreach \i in {-3.5,-1.5,0.5,2.5,4.5,6.5,8.5,10.5,12.5,14.5,16.5}
{
\rhoO{-3.5+\i}{-1}
}
 \foreach \i in {0,...,3}{
\MYcircle{\i-.5+8}{3.5-\i}
\MYcircle{\i-4-.5}{0.5+\i}
}
 \MYcircle{-5.5}{-.5}  \MYcircle{-6.5}{-.5}   \MYcircle{-7.5}{-.5} \MYcircle{11.5}{-.5} \MYcircle{12.5}{-.5}\MYcircle{13.5}{-.5}
\end{tikzpicture}.
\label{eq:NESS}
\end{align}
\ew
Then we simplify the region which is farther than $t$ away from the sub-system, by using  $\lim_{L\to \infty} \Tr E(0)^L= \ket{\triangle}\bra{\msquare}$. In the last step we use left and right unitality (iii), (iv) and conditions~\eqref{eq:cond0State} to simplify the diagrams from the left and the right and obtain~\eqref{eq:NESS3} if $t>|A|/2$.

It is natural to wonder what kind of steady states (c.f.~\eqref{eq:NESS3}) the 3-way evolution prepares in finite time. This question is hard to answer analytically. As an exploratory analysis, we performed numerical simulations on systems of up to $14$ qubits. Since we lack a thorough understanding of the 3-way set of channels, we focused on two classes. The first one is constructed by considering random perturbations around a completely depolarising channel. All of the chosen perturbations are postselected such that the superoperators are completely positive, which in practice means that the perturbations have Frobenius norm $\lesssim 0.7$.
The second class is a weakly perturbed dual-unitary quantum channel, where only the components that break unitality are perturbed. The elements\footnote{Except for the element mapping the identity operator to itself.} of the initial gate $q$ were slightly shrunk to maintain complete positivity.

To characterise the steady states we looked at their purity, mutual information, and negativity. In the above-mentioned examples we observed that the purity was low and decreased with system size, but was nevertheless higher than in the infinite temperature state $\ket{\mcirc \ldots \mcirc}$. For mutual information and negativity, we considered a bipartition by halving the system. Negativity was zero and mutual information was small, suggesting that the steady states are thermal-like \emph{area law} states. Low mutual information can be perhaps expected based on the exact expressions for correlation functions (cf. Eq.~\eqref{eq:quench6}), which generically decay exponentially with the separation of the two points. 

Even though our exploratory investigation suggests that the resulting steady states are typically almost thermal, it does not exclude more exotic ones. It remains an interesting open question if by fine-tuning the local channel or by going to higher local dimension $d$, one can prepare steady states with qualitatively different properties than the ones witnessed.


\section{Classification of Solvable Initial Mixed States}
\label{sec:MPDO}

In this section we generalise the classification of the so-called solvable initial states arising in dual-unitary dynamics~\cite{piroli2020exact}, to encompass space-time unital quantum channels. As discussed earlier, the natural way to extend solvability from MPS to their mixed state analogues, MPDO, is to demand the existence of (nonvanishing) operators $\ket{\msquare}$ and $\ket{\triangle}$ such that Eq.~\eqref{eq:cond0State} holds. Particularly, with graphical calculations in Sections~\ref{sec:corr_fun_quench} and~\ref{sec:noneq_steady_st}, we showed that the solvability condition for the initial states, in combination with the appropriate (space) unitality of the dynamics, allow for the exact calculation of correlation functions and steady states.

We will now provide a classification of those MPDO, under the additional condition that they admit a \textit{local purification}~\cite{verstraete2004matrix,de2013purifications}. That is, we will assume that the MPDO admits a purification that can be written as a translation-invariant MPS (with a constant but possibly different bond dimension) where the purifying system is local.
Our results are captured in two propositions. In Proposition~\ref{prop:solvability} we connect the local purifying tensor with injective MPS, and show that it can be chosen so that it satisfies a condition analogous to space unitality. Using this result, we derive Proposition~\ref{prop:solvable_one_to_one}, in which we show that solvable MPDOs with local purifications are essentially in one to one correspondence with quantum channels. This reduces the classification of solvability for this important class of initial (mixed) states to the well-studied problem of quantum channels.

\subsection{Mixed States with Local Purification and their Classification}

A translation-invariant MPDO over $L$ pairs of sites is in a local purification form if
\begin{align}
    \rho_L (A) = \Tr_{\gamma_1 \ldots \gamma_L} \ket{\Psi_L(A)} \! \bra{\Psi_L(A)} \quad \forall L
\end{align}
where
\begin{align} \label{eq:purifying_mps_def}
  &\ket{\Psi_L(A)} = \nonumber \\ 
  &\!\!\! \sum_{\{i_k^{\mathrm{L}},i_k^{\mathrm{R}},\gamma_k\}} \!\!\!\!\!\!\! \Tr \!\! \left(  A^{(i^{\mathrm L}_1 i^{\mathrm R}_1 \gamma_1)} \ldots A^{(i^{\mathrm L}_L i^{\mathrm R}_L \gamma_L)} \right) \! \ket{i^{\mathrm L}_1 i^{\mathrm R}_1 \gamma_1  \ldots i^{\mathrm L}_L i^{\mathrm R}_L \gamma_L} 
\end{align}
is the purifying MPS\footnote{The purification is defined up to an irrelevant unitary over the purifying system.} and $A^{(i^{\mathrm L}_k i^{\mathrm R}_k \gamma_k)}$ are $\chi \times \chi$ matrices. Above $i_k^{\mathrm L}, i_k^{\mathrm R}$ denote respectively the left and right indices of the physical space for each pair of sites (i.e., integer and half-integer), while $\gamma_k$ corresponds to the local purification space. Using a representation as in Eq.~\eqref{eq:quench1}, the tensor of a MPDO with a local purification can be taken to be
\begin{align} \label{eq:local_purification}
M_A^{(i^{\mathrm L} i^{\mathrm R} j^{\mathrm L} j^{\mathrm R})} = \sum_\gamma A^{(i^{\mathrm L} i^{\mathrm R}\gamma)} \otimes [A^{(j^{\mathrm L} j^{\mathrm R} \gamma)}]^* \,,
\end{align}
where we have changed notation $T \mapsto M$ to emphasise that the local purification tensor has additional structure. In graphical notation, this translates to
\be
\frac{1}{d^2} \;
\begin{tikzpicture}[baseline=(current  bounding  box.center), scale=.7]
\draw[very thick] (-1.0 ,0.) -- (3.0,0.);
\rhoO{0}{0}\rhoO{2}{0}
\end{tikzpicture}
= 
\begin{tikzpicture}[baseline=(current  bounding  box.center), scale=.7]
\def\dx{0.15}
\def\dy{0.15}
\draw[ thick] (-4.75+.2+\dx,0.+\dy) --(-0.0+0.2+\dx,0.+\dy);
\draw[ thick] (-3.75+\dx,0.75+\dy) -- (-3.75+\dx,-0.0+\dy);
\draw[ thick] (-3.25+\dx,0.75+\dy) -- (-3.25+\dx,0.0+\dy);
\draw[ thick, fill=myblue, rounded corners=2pt] (-4+\dx,0.25+\dy) rectangle (-2.5+\dx,-0.25+\dy);
\draw[ thick] (-1.75+\dx,0.75+\dy) -- (-1.75+\dx,-0.0+\dy);
\draw[ thick] (-1.25+\dx,0.75+\dy) -- (-1.25+\dx,0.0+\dy);
\draw[ thick, fill=myblue, rounded corners=2pt] (-2+\dx,0.25+\dy) rectangle (-0.5+\dx,-0.25+\dy);
\draw[ thick] (-2.75 ,0.25) -- (-2.75,0.25+.2) --(-2.75+\dx,0.25+\dy+.2) --(-2.75+\dx,0.25+\dy);
\draw[ thick] (-4.75 ,0.) -- (-.0+.2,0.);

\draw[ thick] (-3.75,0.75) -- (-3.75,-0.0);
\draw[ thick] (-3.25,-0.0) -- (-3.25,0.75);
\draw[ thick, fill=myorange0, rounded corners=2pt] (-4,0.25) rectangle (-2.5,-0.25);

\draw[ thick] (-1.75,0.75) -- (-1.75,-0.0);
\draw[ thick] (-1.25,-0.0) -- (-1.25,0.75);
\draw[ thick, fill=myorange0, rounded corners=2pt] (-2,0.25) rectangle (-0.5,-0.25);
\draw[ thick] (-0.75 ,0.25) -- (-0.75,0.25+.2) --(-0.75+\dx,0.25+\dy+.2) --(-0.75+\dx,0.25+\dy);
\end{tikzpicture} \; ,
\ee
where yellow and blue rectangles denote $A$ and $A^*$, and are connected by a line representing the summation over $\gamma$. We will also be assuming that the purifying MPS is normalized in the thermodynamic limit
\begin{align}
\lim_{L \to \infty}  \Tr \left( \ket{\Psi_L(A)} \! \bra{\Psi_L(A)} \right) = 1 \,,
\end{align}
a property that is passed on to the resulting MPDO.

The advantage of the local purification form of a MPDO is that it always yields \textit{by construction} a positive-semidefinite operator. In general, however, not every translation-invariant MPDO admits a local purification valid for all sizes~\cite{de2016fundamental}. On the other hand, resolving whether a (general) tensor yields a positive-semidefinite operator for all sizes is known to be an undecidable problem~\cite{de2016fundamental}.

Following the classification of solvable initial states in Ref.~\cite{piroli2020exact}, it is natural to identify any two MPDO as \textit{equivalent in the thermodynamic limit} if they yield the same expectation values over any observable with support over a finite region. More precisely, two MPDO $\rho_L$ and $\sigma_L$, defined over $L$ pairs of sites, are equivalent if
\begin{align}
    \lim_{L \to \infty} \Tr_{\bar R} \rho_L  =  \lim_{L \to \infty} \Tr_{\bar R} \sigma_L
    \label{eq:equivalence}
\end{align}
where $R$ is any finite region and the bar denotes its complement. In that case, we will also say that the tensors producing the two MPDO are equivalent.

We are now ready to state the first result towards the classification of solvable MPDO. For that, we will need the notion of injectivity, which is central in the theory of MPS~\cite{fannes1992finitely,cirac2021matrix}. A MPS tensor is said to be injective if, when seen as a map from the auxiliary to the physical space, it is injective~\cite{cirac2021matrix}.

\begin{Proposition} \label{prop:solvability}
Suppose $M_{A'}$ defines a MPDO in the local purification form~\eqref{eq:local_purification} that fulfils the solvability condition~\eqref{eq:cond0State}. Then $M_{A'}$ is equivalent to a solvable tensor $M_{A}$, also in local purification form, where the purifying MPS is injective (for sizes large enough), and satisfies
\label{subeq:solvable_mpdo}
\begin{align} \label{eq:kraus_solvable_MPDO}
\sum_{k^{\mathrm R}, \gamma} A^{(i^{\mathrm L} k^{\mathrm R} \gamma)}  A^{(j^{\mathrm L} k^{\mathrm R}  \gamma)\dagger} = \frac{\delta_{i^{\mathrm L},j^{\mathrm L}}}{d} \1_{\chi}
 \end{align}
or, in the folded graphical notation,
\be \label{eq:MPDO_classification}
\sqrt{\frac{d}{\chi}} \;\;
\begin{tikzpicture}[baseline=(current  bounding  box.center), scale=.7]
\def\dx{0.15}
\def\dy{0.15}
\draw[ thick] (-4.4+\dx,0.+\dy) --(-2.1+\dx,0.+\dy);
\draw[ thick] (-3.75+\dx,0.75+\dy) -- (-3.75+\dx,-0.0+\dy);
\draw[ thick] (-3.25+\dx,0.75+\dy) -- (-3.25+\dx,0.0+\dy);
\draw[ thick, fill=myblue, rounded corners=2pt] (-4+\dx,0.25+\dy) rectangle (-2.5+\dx,-0.25+\dy);
\draw[ thick] (-2.75 ,0.25) -- (-2.75,0.25+.2) --(-2.75+\dx,0.25+\dy+.2) --(-2.75+\dx,0.25+\dy);
\draw[ thick] (-4.4 ,0.) -- (-2.1,0.);

\draw[ thick] (-3.75,0.75) -- (-3.75,-0.0);
\draw[ thick] (-3.25,-0.0) -- (-3.25,0.75);
\draw[ thick, fill=myorange0, rounded corners=2pt] (-4,0.25) rectangle (-2.5,-0.25);
\draw[ thick] (-2.1+\dx,0.+\dy) -- (-2.1,0.);
\draw[ thick] (-3.25+\dx,0.75+\dy) -- (-3.25,.75);
\Text[x=-2.55,y=0.8]{$\tiny \gamma$}
\end{tikzpicture}
=
\begin{tikzpicture}[baseline=(current  bounding  box.center), scale=.75]
\rhoO{0}{0}
\MYcircle{0.5}{-.0}\MYcircle{.5}{.5}
\end{tikzpicture}= \begin{tikzpicture}[baseline=(current  bounding  box.center), scale=.75]
\draw[very thick] (0.5,0.5)--(.3,0.7);
\draw[very thick] (0.5,0)--(.3,0);
\MYcircle{0.5}{-.0}\MYcircle{.5}{.5}
\end{tikzpicture}
\;,
\ee
\end{Proposition}

\begin{proof}
Our proof generalises the ideas employed in Ref.~\cite{piroli2020exact} (see Appendix~B therein). However, the proof there relies on certain technical results of completely positive maps, and the latter arise naturally for transfer matrices of \textit{pure} initial states. The key step here is to invoke the purification form, which restores some of this structure but at the cost of the additional purifying system.

Indeed, from the purification~\eqref{eq:local_purification} of $M_{A'}$, it follows that the corresponding space transfer matrix $E_{M_{A'}} \coloneqq E_{M_{A'}}(0)$
can be decomposed as 
\begin{align}
 E_{M_{A'}} =  \sum_{i^{\mathrm L} i^{\mathrm R} \gamma} A^{(i^{\mathrm L} i^{\mathrm R} \gamma)} \otimes [A^{(i^{\mathrm L} i^{\mathrm R} \gamma)}]^* \,.   
\end{align}
\noindent As such, it corresponds to (the vectorised form of) the completely positive map 
\begin{align}
    \mathcal E_{M_{A'}} (\cdot) =  \sum_{i^{\mathrm L} i^{\mathrm R} \gamma} A^{(i^{\mathrm L} i^{\mathrm R} \gamma)} (\cdot) [A^{(i^{\mathrm L} i^{\mathrm R} \gamma)}]^\dagger\,.
\end{align}
Due to the normalisation~\eqref{eq:mpdo_normalisation} of the MPDO, $E_{M_{A'}}$ has a unique largest eigenvalue $\lambda = 1$ (and no other eigenvalue of the same modulus). This fact, together with $\mathcal E_M$ being completely positive, imply that the eigenoperator $Z$ satisfying $\mathcal E (Z) = Z$ (i.e., the unique fixed point) is positive-semidefinite~\cite{perez2007matrix}. In addition, the purifying MPS $\ket{\Psi_L(A)}$ is injective (for large enough number of sites $L$) if and only if $Z$ is positive-definite~\cite{sanz2010quantum,piroli2020exact}. Notice also that from the solvability condition~\eqref{eq:cond0State} we can identify $\ket{\triangle} = \ket{Z}$.

However, up to the equivalence in the thermodynamic limit defined earlier 
(cf. Eq.~\eqref{eq:equivalence}), one can replace $A'$ with a tensor $A''$ such that the new purifying MPS $\ket{\Psi_L(A'')}$ is injective (for $L$ large enough). It is of smaller or equal bond dimension, and the corresponding fixed point $Z$ has full rank. This result follows directly from the steps of the proof in~\cite{piroli2020exact}.

Now, since $Z$ can be taken to be positive-definite, we are able to define $A^{(i^{\mathrm L} i^{\mathrm R} \gamma)} =  Z^{-1/2} A''^{(i^{\mathrm L} i^{\mathrm R} \gamma)} Z^{1/2}$, which amounts to a gauge transformation in the auxiliary space of the tensor $A''$. Crucially, this transformation enforces that  $E_{M_{A}} \ket{\mcirc} = \ket{\mcirc}$. Due to the uniqueness of the fixed point of the transfer matrix, this also leads to Eq.~\eqref{eq:kraus_solvable_MPDO}, while at the same time $M_{A}$ is clearly in a local purification form.
\end{proof}

The above reduces to the solvable MPS of Ref.~\cite{piroli2020exact} when the purifying space trivializes. However, the existence of a nontrivial purifying space qualitatively changes the classification.

To show that, we need to reinterpret Eq.~\eqref{eq:kraus_solvable_MPDO} by regrouping the indices of the tensor $A$. Treating jointly the left bond and the left physical index as input, and the corresponding right indices as output, and assuming as usual a fixed reference basis, defines the operators
\begin{subequations} \label{subeq:mpdo_channel}
\begin{align} \label{eq:Kraus_of_isometry}
K^{(\gamma) \dagger} \coloneqq \sum_{i^{\mathrm L} j^{\mathrm R}} A^{(i^{\mathrm L} j^{\mathrm R} \gamma)} \otimes \ket{i ^{\mathrm L}} \! \bra{j^{\mathrm R}}  \,.
\end{align}
\noindent In this notation, the meaning of Eq.~\eqref{eq:kraus_solvable_MPDO} becomes more transparent, since it enforces that
\begin{align} \label{eq:kraus_normalization}
    \sum_{\gamma} K^{(\gamma) \dagger} K^{(\gamma) } = \frac{\1_{d \chi}}{d} \,,
\end{align}
\end{subequations}
\noindent i.e., that $K^{(\gamma)}$ form a valid Kraus operator set of a CPTP map, up to an irrelevant normalization factor. Said differently, it is a valid \emph{quantum channel} in the combined auxiliary and real space.
Here $\chi$ denotes the bond dimension of $A$, which places a bound on the Kraus rank of the quantum channel\footnote{Proposition~\ref{prop:solvability} may only decrease the bond dimension of the original MPS $A'$.}.

We have thus observed that Eq.~\eqref{eq:Kraus_of_isometry} establishes a mapping from tensors $A$ satisfying the conditions of Proposition~\ref{prop:solvability} to quantum channels. As with any purification, a transformation
\begin{align}
    A^{(i^{\mathrm L} i^{\mathrm R} \gamma)} \mapsto \sum_{\gamma'} V_{\gamma \gamma'} A^{(i^{\mathrm L} i^{\mathrm R} \gamma')}
\end{align}
for $V$ isometry leaves the MPDO and the quantum channel invariant, and only affects the form of the Kraus decomposition. Moreover, notice that quantum channels resulting from Proposition~\ref{prop:solvability} have additional structure since the purifying MPS always can be taken to be injective. It is convenient to rephrase this injectivity condition in terms of the Kraus decomposition; then the map
\begin{align} \label{eq:channel_injectivity}
    X \mapsto \sum_{\gamma} \Tr_\chi \left[  ( X \otimes \1_d ) K^{(\gamma)} \right] \ket{\gamma} \quad \text{is injective,}
\end{align}
where $X$ acts in the auxiliary space and the partial trace is also with respect to the auxiliary space. It is not hard to see that this is a property of the quantum channel and does not depend on the particular Kraus representation\footnote{This is since every different Kraus representation of a quantum channel is obtained by the action of an isometry over the space of the index (i.e., here the purification space), and such action cannot modify injectivity.}.

With this observation, we can establish that not only do solvable tensors yield \emph{quantum channels}, but also the opposite holds.

\begin{Proposition} \label{prop:solvable_one_to_one}
Up to equivalence in the thermodynamic limit, to every tensor in purification form yielding a solvable MPDO corresponds a quantum channel satisfying the injectivity property~\eqref{eq:channel_injectivity}. Conversely, to every quantum channel satisfying~\eqref{eq:channel_injectivity} corresponds a solvable MPDO in purification form.
\end{Proposition}
\begin{proof}

We have already shown the direct implication, which follows from Proposition~\ref{prop:solvability}. The quantum channel can be constructed via Eqs.~\eqref{subeq:mpdo_channel}.

For the converse, the first step is decomposing the quantum channel into a valid (minimum) set of Kraus operators. Then from Eq.~\eqref{eq:Kraus_of_isometry}, one can deduce the family of operators $A'^{(i^{\mathrm L} j^{\mathrm R} \gamma)}$ (see Eq.~\eqref{eq:Kraus_to_MPDO}). In turn, the above family of operators yields a MPS $\ket{\Psi_L(A)}$ via Eq.~\eqref{eq:purifying_mps_def}. Notice that different Kraus decompositions will result in the action of a unitary over the purification space, thus producing different purification of the same MPDO. Importantly, the state is guaranteed to be normalized in the thermodynamic limit. This is since
\begin{align}
    \lim_{L \to \infty}  \braket{ \Psi_L(A) | \Psi_L(A)} = \lim_{L \to \infty} \Tr \left( E_{M_{A'}} ^L \right) = 1
\end{align}
where the last equality follows from the injectivity condition; the latter ensures that the largest eigenvalue (in magnitude) of the transfer matrix is unique and equal to the spectral radius of the quantum channel. As a result, we have obtained a valid MPDO in purification form, which satisfies Eq.~\eqref{eq:kraus_solvable_MPDO}.
\end{proof}

In summary, the utility of the above result is that quantum channels (satisfying a generic injectivity property) are mapped to solvable MPDO; inverting Eq.~\eqref{eq:Kraus_of_isometry} one gets
\begin{align} \label{eq:Kraus_to_MPDO}
    A^{(i^{\mathrm L} j^{\mathrm R} \gamma)} = \braket{i^{\mathrm L} | K^{(\gamma) \dagger} | j^{\mathrm R}} \,.
\end{align}
Up to equivalence, this is exhaustive for MPDO with local purification.

\subsection{Example: Solvable MPDO with $\chi = 1$}

Let us consider the simplest case of bond dimension $\chi = 1$. This instance corresponds to solvable MPDO that completely factorize over two sites $\rho =  \rho_0 ^{\otimes L}$ and are characterised by quantum channels acting solely over $\mathbb C ^{d}$. In particular, 
\begin{align}
    \rho_0 = \!\!\!\!\!\sum_{ i^{\mathrm L} j^{\mathrm R} k^{\mathrm L} l^{\mathrm R} \gamma } 
    \!\!\!\!\!\braket{i^{\mathrm L} | K^{(\gamma) \dagger} | j^{\mathrm R}} \braket{l^{\mathrm R} | K^{(\gamma) } | k^{\mathrm L}} \ket{i^{\mathrm L}  j^{\mathrm R}} \! \bra{k^{\mathrm L}  l^{\mathrm R}} .
\end{align}
Notice that injectivity condition is trivially satisfied. In addition, normalization~\eqref{eq:kraus_normalization} ensures that $\Tr (\rho_0) = 1$, hence the global MPDO is exactly normalized for all sizes.

For instance consider the family of depolarising channels 
\be
\Phi(\sigma) = \alpha \sigma + (1-\alpha) \1
\ee
which connects the totally mixed initial state $\ket{\mcirc \dots \mcirc}$ at $\alpha=0$ to pure initial state of Bell pairs for $\alpha=1$. The first was effectively the initial state in the first calculation of correlation functions on the infinite temperature initial states \cite{bertini2019exact}, whereas the second is the only (up to local unitaries) pure solvable initial state for $\chi=1$~\cite{piroli2020exact}.

\section{Conclusions and perspectives}
\label{sec:conclusions}

In this work, we generalise the ideas of dual-unitarity to noisy quantum circuits. The key is to impose that the evolution along one or both of the space directions, and possibly backward in time, is a valid quantum channel. This amounts to different unitality conditions, which give rise to various solvable classes of dynamics.
Via dimension counting arguments we show that the resulting families of models are much larger than averaging of dual-unitary gates. Regarding their structure, we show that all 4-way unital gates (i.e., unital along the space and time directions) can be expressed as affine combinations of a restricted subclass of dual-unitary ones. This is of particular relevance, since a characterisation of dual-unitary gates beyond qubits is lacking. On the other hand, the problem seems to be considerably more tractable when the class of dynamics is enlarged, as considered here. When these affine combinations can be expressed as convex ones is an interesting question left for the future.
Moreover, we show how the additional sideways unitality conditions enable simplifications of the objects appearing in the calculation of physical properties, for instance, correlation functions. As a simple but physically relevant example, we show how the average evolution of noisy circuits can lead to such simplifications and exact solvability, even if dual-unitarity of the gates is absent. We also address the question of a solvable quantum quench, that is, the evolution of a mixed state that allows for exact calculations. We provide necessary condition for the initial state, in the form of a Matrix-Product Density Operator (MPDO), to be solvable. Focusing MPDO with a local purification, we show that they are essentially in one-to-one correspondence with quantum channels.

The proposed classes of evolutions lead to exact solutions, which provide insights into the challenging dynamics of noisy quantum circuits. Quantum circuits of quantum channels model for instance quantum computation, which is not totally isolated from the environment. As such, we expect the present ideas to also be relevant in that context.

When using the simplifications for exact calculations, we focused on the homogeneous dynamics. Nevertheless, the simplifying properties can also be used in evolutions, which are disordered in space or time, as long as the additional properties hold for all of the local channels. A detailed study of these effects is left for future work. 

One can also extend the discussion to $k$-local operators. After simplifications from 3-way unitality, their correlations would be expressed using $k$ qudit quantum channels. Further characterisation of these channels (also for $k=1$) is an interesting direction left for the future.

One would wish to study also objects beyond correlators, such as purity, entanglement, negativity, and so on. With dual-unitarity this proved possible, but here the analogous calculations are not straightforward. The main obstruction is that while folded dual-unitary gates $q=U \otimes U^*$ are unitary, 4-way unital channels are not. In particular, $q q^\dagger$ and $q^\dagger q$ do not simplify. We expect one needs to satisfy additional conditions, in order to circumvent this difficulty.

\section*{Acknowledgements}
We thank Toma\v z Prosen, Bruno Bertini, and  Ignacio Cirac for fruitful discussions and Lorenzo Piroli, Felix Fritzsch, Toma\v z Prosen, Bruno Bertini for comments on the manuscript. Both authors are supported by the Alexander von Humboldt Foundation. We acknowledge support by the DFG (German Research Foundation) under Germany's Excellence Strategy -- EXC-2111 -- 390814868. 
\vspace{.2cm}

\bibliographystyle{quantum}
\bibliography{bibliography}
\appendix

\section{Proofs of Propositions} \label{sec:app_proofs}

Here we report the proofs omitted from the main text.

\subsection{Proof of Proposition~\ref{prop:counting_dimensions}}

Let $\{ F_k \}_{k=0}^{d^2-1}$ be a (Hilbert-Schmidt) orthonormal basis of hermitian operators acting over $\mathbb C^{d}$. Without loss of generality we take $F_0 \equiv \1_d / \sqrt{d}$, hence the rest of the operators are traceless. In turn, this induces a basis in the space of hermiticity-preserving super-operators over $\mathbb C^{d} \otimes \mathbb C^{d}$, defined through its action
\begin{align}
    \mathcal F_{kl,ij} (X_A \otimes X_B) = F_k \Tr(X_A F_i) \otimes F_l \Tr(X_B F_j) \,,
\end{align}
where in this notation $\mathcal T = \mathcal F_{00,00}$ and $X_A,X_B$ are the operator acting on $A,B$. Then, quantum channels can be expanded in the form
\begin{align} \label{eq:expansion_Pauli}
    \mathcal E = \mathcal T + \sum_{ijkl} \alpha_{kl,ij} \mathcal F_{kl,ij}
\end{align}
where the $\alpha$ coefficients are real.

Let us start with the case of perfect channels. The key observation is that, in the above expansion, a completely positive map is a perfect channel if and only if (for the nonzero coefficients) at most one of the $(ijkl)$ indices is zero, i.e., corresponds to the identity operator. In other words, any nonvanishing component with at least two nonzero indices $(ijkl)$ will necessarily imply the failure of at least one of the conditions (i)-(vi). This is readily verified graphically using the folded picture notation. 

For instance, assume condition (iii). Then, for at least one of $j,l$ nonzero, 
\be
\alpha_{0l,0j} =
\begin{tikzpicture}[baseline=(current  bounding  box.center), scale=.7]
\Wgategreen{0}{0}
\MYcircle{-.5}{-.5}
\MYcircle{-.5}{0.5}
\MYcircleB{.5}{-.5}
\MYcircleB{.5}{0.5}
\end{tikzpicture}=
\begin{tikzpicture}[baseline=(current  bounding  box.center), scale=.7]
\draw[very thick](-.5,-.5)--(0.5,-.5);
\draw[very thick](-.5,.5)--(0.5,.5);
\MYcircle{-.5}{-.5}
\MYcircle{-.5}{0.5}
\MYcircleB{.5}{-.5}
\MYcircleB{.5}{0.5}
\end{tikzpicture}
= 0 ,
\ee
where the black circles correspond operators $ F_j$ and $ F_l$.

As a result, the allowed nonzero $\alpha$ coefficients are only of the following types:
\begin{itemize}
    \item  None of the $(ijkl)$ indices is zero [$(d^2-1)^4$ operators].
    \item Exactly one of the $(ijkl)$ indices is zero [$4 \cdot (d^2-1)^3$ operators].
\end{itemize}

A priori, it is not obvious that all the preceding linearly independent deformations of $\mathcal T$ can produce completely positive maps for a suitable range of the parameters. However, this is indeed the case for small enough (but nonzero) values for all of the specified $\alpha$ parameters. This follows by considering the corresponding CJS $J(\mathcal E)$. Notice that 
$J(\mathcal F_{kl,ij}) = F_k \otimes F_l \otimes F_i \otimes F_j$ (see Eq.~\eqref{eq:choi_graphical}) and thus $J(\mathcal T) = \1/d^2$. As a result, small enough (hermitian) perturbations around $J(\mathcal T)$ preserve positivity.

In conclusion, by considering perturbations around $\mathcal T$, we have shown that there exists an affine subspace having real dimension
\begin{align*}
    (d^2-1)^4 + 4 \cdot (d^2-1)^3 = (d^4-1)^2- 4(d^2-1)^2 \,,
\end{align*}
which contains the convex set of perfect quantum channels. This affine subspace is clearly a minimal one, as it follows from the completeness of the expansion~\eqref{eq:expansion_Pauli} and the previous considerations of the allowed indices.

The proof for the remaining cases proceeds in the same fashion, with a more refined counting.
In the case of 4-way unital channels, one needs to include additionally operators of the form $\mathcal F_{kl,ij}$ such that
\begin{itemize}
    \item Either $i=l=0$ or $j = k =0$ and the rest of the indices are nonzero [$2 \cdot (d^2-1)^2$ operators].
\end{itemize}
For 3-way (say, TU and LU), also include:
\begin{itemize}
    \item $j = l =0$ and the rest of the indices are nonzero [$(d^2-1)^2$ operators].
\end{itemize}
The counting is identical for the rest of the 3-way cases. Finally, for 2-way (say, LU) one needs to add to the above:
\begin{itemize}
    \item $i = j =0$ and the rest of the indices are nonzero [$(d^2-1)^2$ operators].
    \item $i = j = l = 0$  and the remaining index is nonzero [$(d^2-1)$ operators].
\end{itemize}
For the TU case, instead add to the 3-way:
\begin{itemize}
    \item $i = k =0$ and the rest of the indices are nonzero [$(d^2-1)^2$ operators].
\end{itemize}

Performing the sums, one arrives at the values reported in Table~\ref{tab:dimensions}.

\subsection{Proof of Proposition~\ref{prop:qubits_affine_dual}}
\label{app:proofProp2}

Let us first restate some facts from the main text for convenience. In~\cite{claeys2021ergodic} it was shown that
\be \label{eq:dual_unitary_param}
   \!\! U_k (J_k) = (W^k _1 \otimes W^k _2)  \exp \left(i J_k s_3 \otimes s_3\right) S (V^k _1 \otimes V^k _2),
\ee
are dual-unitary gates for any local dimension $d$, where $J_k\in \mathbb R$, the local unitaries are arbitrary, $S$ denotes the swap operator, and $k$ is an index that we will use shortly to distinguish between different dual-unitaries. The spin matrices $s_i$ with $i  = 1,2,3$ satisfy the $su(2)$ commutation relations
\begin{align} \label{eq:su2_algebra}
    [ s_i , s_j ] = i \sum_k \epsilon_{ijk} s_k \,.
\end{align}
%
To prevent any potential confusion, we will avoid vectorising in this subsection and work directly with superoperators.

As in the proof of Proposition~\ref{prop:counting_dimensions}, we wish to adopt the notation $F_i$ ($i = 1, \dots, d^2-1$) for a Hilbert-Schmidt orthonormal set of hermitian and traceless operators. The set becomes a basis of hermitian operators over $\mathbb C^d$ by appending $F_0 \equiv \1 / \sqrt{d}$. Without loss of generality, we will henceforth take $F_i = c_d s_i$ for $i=1,2,3$, where $c_d$ is a constant to ensure normalization. This is always possible since $\braket{s_i,s_j} = 0$.

Firstly, let us show that the completely depolarizing channel $\mathcal T$ (a perfect channel, thus also 4-way unital) can be expressed as an affine combination of dual-unitary gates. To see this, notice that it factorizes $\mathcal T = \mathcal T_A \otimes \mathcal T_B$ and each of the $T_{A/B}$ is unital. Moreover, unital quantum channel can be expressed as affine combinations of unitary gates~\cite{mendl2009unital}. We can thus decompose
\begin{align}
    \mathcal T_A = \sum_{i} \lambda_i V_{1}^i (\cdot) V_{1}^{i \dagger}, \quad \sum_i 
    \lambda_i = 1
\end{align}
and similarly for $\mathcal T_B$. Combining these facts with the trivial observation that swapping the two subsystems after depolarizing has no effect, i.e., $\mathcal T = \mathcal S \mathcal T$ for $\mathcal S = S (\cdot) S$, we reach the conclusion that there exists a decomposition
\begin{align}
    \mathcal T = \sum_{ij} \lambda_{i} \lambda_{j}  S (V_1^i \otimes V_2^j) (\cdot)  (V_1^{i \dagger} \otimes V_2^{j \dagger}) S
\end{align}
which is an affine combination of dual-unitary gates with $J_k = 0$. As an example for $d=2$,
\begin{align*}
    \mathcal T = \mathcal T_A \otimes \mathcal T_B = \frac{1}{16} \sum_{a,b =0}^3 (\sigma_a \otimes \sigma_b)  S (\cdot) S (\sigma_a \otimes \sigma_b),
\end{align*}
where $\sigma_i$ denote the Pauli matrices and $\sigma_0 = \1$.

Following the proof of Proposition~\ref{prop:counting_dimensions}, we expand the quantum channel $\mathcal E$ as
\begin{align} 
    \mathcal E = \mathcal T + \sum_{ijkl} \alpha_{kl,ij} \mathcal F_{kl,ij} .
\end{align}
Therein we have also shown that imposing 4-way unitality over $\mathcal E$ constrains the allowed nonzero $\alpha$ coefficients to be of the following types:
\begin{enumerate}[(a)]
    \item  None of the $(ijkl)$ indices is zero.
    \item Either $i=l=0$ or $j = k =0$ and the rest of the indices are nonzero.
    \item Exactly one of the $(ijkl)$ indices is zero.
\end{enumerate}
The rest of the proof consists of constructing decompositions
\begin{align} \label{eq:expansion_dual_unitary}
    \mathcal F = \sum_k \lambda_k U_k(J_k) (\cdot) U_k(J_k)^\dagger , \; \sum_k \lambda_k = 0 , \; \lambda_k \in \mathbb R 
\end{align}
for every $\mathcal F$ belonging in one of the allowed three classes.


We will show momentarily that elements in classes (a) and (b) can be constructed with $J_k = 0$, i.e., by varying only the local part of the dual unitaries~\eqref{eq:dual_unitary_param} and the coefficients $\lambda_k$. On the other hand, for elements in class (c) we will invoke a nontrivial $J_k$. 
\\ \\
\textbf{(a)} Consider a single qudit and for $k,l \ge 1$ define the superoperator
\begin{align}\label{eq:Gkk}
    \mathcal G_{kl} & \equiv F_k \Tr[(\cdot) F_l]  \,.
\end{align}
Since every unital channel can be expressed as an affine combination of unitaries, then necessarily $\mathcal G_{kl}$ admits a decomposition over unitary gates with the corresponding coefficients satisfying $\sum_i \lambda_i = 0$. Such decompositions can also be easily made explicit via unitary 1-designs~\cite{bengtsson2017geometry}. Based on that, we can now directly write
\begin{align*}
    \mathcal F_{kl,ij}  = \mathcal S (\mathcal G_{li} \otimes \mathcal G_{jk}) 
\end{align*}
which, by the previous consideration, results in a linear combination of the form~\eqref{eq:expansion_dual_unitary} with $J_k=0$ dual-unitary gates.

For instance, in the case of qubits,
\begin{align}
    \mathcal G_{kk}  = \frac{1}{2} \big[ \sigma_0 (\cdot) \sigma_0 + \sigma_k (\cdot) \sigma_k - \sum_{\substack{a \ne k \\ a \ge 1}} \sigma_a (\cdot) \sigma_a \big] \;,
\end{align}
while the more general case $\mathcal G_{kl}$ can be obtained by an additional conjugation with a unitary.
\\ \\
\textbf{(b)} Using a similar reasoning, for $i,l \ge 1$ we have 
\begin{align*}
    \mathcal F_{0l,0j}  = \mathcal S ( \mathcal G_{lj} \otimes \mathcal T_B)
\end{align*}
and 
\begin{align*}
    \mathcal F_{k0,i0} = \mathcal S ( \mathcal T_A \otimes \mathcal G_{ki} )
\end{align*}
which, using the expansions above, gives rise to form~\eqref{eq:expansion_dual_unitary}, again with $J_k = 0$.
\\ \\
\textbf{(c)} Firstly, let us illustrate the idea behind the construction with an explicit example for qubits. We define
\begin{align*}
    V_{33} =  \1 \otimes \1 + 2 i \epsilon \sigma_3 \otimes \sigma_3  ,
\end{align*}
which corresponds to infinitesimal rotations generated by the nonlocal part of Eq.~\eqref{eq:dual_unitary_param} ($J_k \propto \epsilon \ll 1$). In particular, we can form the linear combination of (adjointly acting) dual-unitary gates
\begin{align*}
    & \;\quad V_{33}  S (\cdot) S V_{33}^\dagger  - S (\cdot) S  \nonumber \\ &= 2i \epsilon [(\sigma_3 \otimes \sigma_3) S (\cdot) S -  S (\cdot) S (\sigma_3 \otimes \sigma_3)] + \mathcal{O}(\epsilon^2)
    .
\end{align*}
As it can be verified from the last equation, we can map
\begin{align*}
    \sigma_3 \otimes \sigma_2 \mapsto  2 \epsilon \sigma_1 \otimes \1 \; ,
\end{align*}
which allows turning a single Pauli to an identity. This is exactly the required ingredient for class (c), which stems from the nonlocal part of the dual-unitary gate. 

For arbitrary $d$, since $F_3 \propto s_3$, we define
\begin{align*}
     V_{33} \equiv \1 \otimes \1 + i \epsilon F_3 \otimes F_3
\end{align*}
and
\begin{align}\begin{split}
    \mathcal J  &\equiv V_{33}  S (\cdot)S V_{33}^\dagger  -  S(\cdot)S  
    \\ &=i \epsilon [(F_3 \otimes F_3),   S(\cdot)S  ] + \mathcal{O}(\epsilon^2) \,.\end{split}
\end{align}
Notice that since $V_{33} S$ is dual unitary up to order $\epsilon^2$.
$\mathcal J $ maps traceless operator $s_3 \otimes s_2$ to a non-traceless one:
\begin{align*}
    s_3 \otimes s_2 \mapsto   \epsilon s_1 \otimes s_3^2 \; .
\end{align*}
Crucially, $\Tr (s_3^2) > 0 $, so up to linear combinations from (a), we obtained $F_{10,32}$. We combine this with single qudit rotations to obtain all elements  $\mathcal F _{l0,ij}$.
In particular, we first map $F_i$ to $F_3$ and $F_j$ to $F_2$, apply $\mathcal J$ and map $F_1$ to $F_l$. Therefore, the following superoperator
\begin{align*}
    \mathcal L_{l, ij} \equiv ( \mathcal G_{l 1} \otimes \mathcal \1 ) \mathcal J  (\mathcal G_{3i} \otimes \mathcal G_{2 j})
\end{align*}
gives us $\mathcal F _{l0,ij}$ up to linear combinations of elements from (a).

We are left to show that $\mathcal L_{l,ij}$ is a linear combination of dual-unitary gates of the form~\eqref{eq:expansion_dual_unitary} corresponding to an infinitesimal rotation of the nonlocal part, with varying local unitaries. 
It is proportional to the difference of the unitary channel obtained from the dual-unitary matrices
\begin{align*}
&(W_1 \otimes W_2)  \exp \left(i \epsilon s_3 \otimes s_3\right) S (V_1 \otimes V_2),\\
&(W_1 \otimes W_2) S (V_1 \otimes V_2),
\end{align*}
where the single qudit gates are specified by $V_1 (F_i)  V_1^\dagger= F_3$, etc. 

The rest of the $\mathcal F$ operators are constructed similarly (e.g., through Hilbert-Schmidt adjoints). 
Mark that the above construction fails for components forbidden by the expansion of 4-way unital quantum channels, e.g., for $\mathcal F _{00,jl}$.

\subsection{Proof of Proposition~\ref{prop:qubit_noise} }
\label{app:averaging}

The quantum channel $\mathcal E_p$ defined in Eq.~\eqref{eq:convex_channel} is by construction a convex combination of unitaries, thus properties (i), (ii) follow automatically.

To show the space unitality properties (iii) and (iv), we will first evaluate the exponential in the standard decomposition~\eqref{eq:standard_decomposition}. This results to 
\begin{align} \label{eq:app:standard_decomp2}
 U_{\lambda} = W_1 \otimes W_2  \left( \sum_{j=0}^3 \frac{\alpha_j}{2} \sigma_j \otimes \sigma_j \right) V_1 \otimes V_2 
\end{align}
where the parameters $\alpha_j$ satisfy
\begin{subequations} \label{eq:app:trig_coeff}
\begin{align}
    \!\!\!\frac{1}{4} \left| \alpha_0 \right|^2 \!&=\!  \prod_i \cos^2\theta_i + \prod_i \sin^2\theta_i  , \\
    \!\!\!\frac{1}{4} \left| \alpha_j \right|^2  \!&= \! \sin^2\theta_j \!\prod_{i\ne j}\! \cos^2\theta_i \!+\! \cos^2\theta_j \!\prod_{i \ne j}\! \sin^2\theta_i ,  j >0 .
\end{align}
\end{subequations}
Remember that $\lambda$ captures all parameters $\theta_i$ and also the parameters of the single site gates $W_{\chi}, V_{\chi}$.

Our strategy will be the following. First we will show that space unitality of $\mathcal E_p$ can be rephrased as
\begin{align} \label{eq:app:space_unital_noise}
 \int  d\lambda  \, p(\lambda) \sum_{j=0}^3 \frac{\left| \alpha_j \right|^2}{4}  W_\chi \sigma_j V_\chi (\cdot)  V^{\dagger}_\chi \sigma_j W^\dagger_\chi  =  \mathcal T 
\end{align}
%
%
where $\chi = 1,2$ corresponds to right and left unitality, respectively, and $\mathcal T =  \frac{\1}{2} \Tr (\cdot)$ is the completely depolarizing channel. Note that this does not mean that the averaged channel $\mathcal E_p$ is completely depolarising.
Then we will prove the validity of Eq.~\eqref{eq:app:space_unital_noise} using our two assumptions. In particular, using Assumption~1, we can rewrite~\eqref{eq:app:space_unital_noise} in vectorised notation as
\begin{align} \label{eq:app:space_unital_noise2}
&\!\!\sum_{j = 0}^3 \sigma_j \otimes \sigma_j^*  = \\ 
&\!\!= \! \int \! d\lambda_W d\lambda_V  \, p_{WV}(\lambda_W,\lambda_V)\! (W^*_\chi \otimes W_\chi)  Q  (V_\chi \otimes V^{*}_\chi) \notag
\end{align}
where
\begin{align}
    Q \!\equiv \!\int\! d \theta_1 d\theta_2 d\theta_3 \, p_{\theta}(\theta_1,\theta_2,\theta_3) \left( \sum_{j=0}^3 \left| \alpha_j \right|^2 \sigma_j \otimes \sigma_j \right)\!,
\end{align}
and we used the fact that $\mathcal T = \frac{1}{4} \sum_{j = 0}^3 \sigma_j (\cdot) \sigma_j$.
Subsequently, we will show that Assumption~2 guarantees that
\begin{align} \label{eq:app:q}
    Q = \sum_{j = 0}^3 \sigma_j \otimes \sigma_j^* \,.
\end{align}
This fact implies that $Q$ itself is proportional to a (vectorised) depolarizing channel, and thus $(W \otimes W^*) Q (V \otimes V^*) = Q$ for arbitrary unitaries $V,W$. As a result, Eq.~\eqref{eq:app:space_unital_noise2} holds true for both $\chi = 1,2$, regardless of the distribution $p_{WV}$.

To conclude, let us now show the two missing pieces of the proof.
\newline

\subsubsection*{Proof that right/left unitality of $\mathcal E_p$ reduces to Eq.~\eqref{eq:app:space_unital_noise}}

The condition for right unitality of $\mathcal E_p$ can be graphically represented as
\begin{align}
\int
d{\lambda} p({\lambda})
\begin{tikzpicture}[baseline=(current  bounding  box.center), scale=.7]
\def\dx{0.55}
\def\dy{0.55}
\draw[ thick] (-4.75+\dx,0.95+\dy) -- (-4.75+\dx,-0.95+\dy);
\draw[ thick] (-3.75+\dx,-0.95+\dy) -- (-3.75+\dx,0.95+\dy);
\draw[ thick, fill=myred, rounded corners=2pt] (-5+\dx,0.35+\dy) rectangle (-3.5+\dx,-0.35+\dy);
\Text[x=-4.25+\dx,y=\dy]{${U_{ \lambda }}^*$}
\draw[ thick] (-4.75,0.95) -- (-4.75,-0.95);
\draw[ thick] (-3.75,-0.95) -- (-3.75,0.95);
\draw[ thick, fill=myblue, rounded corners=2pt] (-5,0.35) rectangle (-3.5,-0.35);
\Text[x=-4.25,y=0]{$U_{ \lambda }$}
\draw[ thick] (-3.75,0.95) -- (-3.75+\dx,0.95+\dy);
\draw[ thick] (-3.75,-0.95) -- (-3.75+\dx,-0.95+\dy);
\end{tikzpicture}
 = 
 \begin{tikzpicture}[baseline=(current  bounding  box.center), scale=.7]
\def\dx{0.55}
\def\dy{0.55}
\draw[ thick] (-3.75,0.95)--(-3.75,0.5) -- (-3.75+\dx,0.5+\dy)--(-3.75+\dx,0.95+\dy);
\draw[ thick] (-3.75,-0.95)--(-3.75,-0.5) -- (-3.75+\dx,-0.5+\dy)--(-3.75+\dx,-0.95+\dy);
\end{tikzpicture}
\end{align}
Using the standard decomposition~\eqref{eq:app:standard_decomp2} for $U_\lambda$, the LHS of the graphical equation above evaluates to
\begin{align*}
      &\int d\lambda \, p(\lambda) 
      \\
      &\sum_{j,k = 0}^3 \frac{\alpha_j \alpha_k^*}{4} (V_1 \otimes V_1^*) \sigma_j \otimes \sigma_k^* (W_1 \otimes W_1^*) \Tr(\sigma_j \sigma_k) \\ 
      &= \int d\lambda \, p(\lambda) \sum_{j = 0}^3 \frac{\left| \alpha_j \right|^2 }{2} (V_1 \otimes V_1^*) \sigma_j \otimes \sigma_j^* (W_1 \otimes W_1^*) 
\end{align*}
where in the first step we used the graphical rule
\begin{align}
    \begin{tikzpicture}[baseline=(current  bounding  box.center), scale=.7]
\draw[ thick,  rounded corners=2pt] (0,0) rectangle (1,1);
\Text[x=.5,y=.5]{$\sigma_j$}
\draw[ thick,  rounded corners=2pt] (2,0) rectangle (3,1);
\Text[x=2.5,y=.5]{${\sigma_k}^*$}
\draw[ thick] (.5,1) -- (.5,1.5) -- (2.5,1.5)--(2.5,1);
\draw[ thick] (.5,0) -- (.5,-.5) -- (2.5,-.5)--(2.5,0);
\end{tikzpicture} = \Tr(\sigma_j \sigma_k) = 2 \delta_{jk} \,.
\end{align}
The analogous condition holds for left unitality by replacing $W_1 \mapsto W_2$ and $V_1 \mapsto V_2$. Passing from the vectorised notation used here to superoperator notation, ones recovers Eq.~\eqref{eq:app:space_unital_noise}.

\subsubsection*{Proof that Assumption~2 implies Eq.~\eqref{eq:app:q}}

We need to show that
\begin{align} \label{eq:app:betas}
     \int d\theta_1 d\theta_2 d\theta_3 \left| \alpha_j \right|^2 = 1  \quad \forall j = 0,1,2,3 \,.
\end{align}
Substituting $\theta_j = \pi/4 + \delta_j$ ($j = 1,2$) one gets
\begin{widetext}
\begin{subequations}
\begin{align}
    \left| \alpha_0 \right|^2 &=  (1 - \sin 2 \delta_1 ) (1 - \sin 2 \delta_2) \cos^2  \theta_3  +  (1 + \sin 2 \delta_1) (1 + \sin 2 \delta_2) \sin^2  \theta_3   \\
    \left| \alpha_1 \right|^2 & =   \sin 2 \delta_1 \cos 2 \theta_3 - \sin 2 \delta _2 \left( \sin 2 \delta _1 + \cos
   2 \theta_3 \right) + 1   \\
   \left| \alpha_2 \right|^2 &=  \left(\sin 2 \delta_2 -\sin 2 \delta_1\right) \cos 2 \theta_3 - \sin 2 \delta _1  \sin 2 \delta_2  + 1    \\
   \left| \alpha_3 \right|^2 &=  \left(1 - \sin 2 \delta_1 \right) \left( 1 - \sin 2 \delta _2 \right) \sin ^2 \theta _3 + \left( 1 + \sin 2 \delta _1 \right) \left( 1 + \sin 2 \delta_2 \right) \cos ^2 \theta_3   \,.
\end{align}
\end{subequations}
\end{widetext}
From the above expressions it is easy to see that Eqs.~\eqref{eq:app:betas} are all satisfied if
\begin{subequations}
\begin{align}
   \! \!\int\!\! d\delta_1 d\delta_2 d\theta_3 p_{\theta} (\frac{\pi}{4}\!+\! \delta_1,\frac{\pi}{4} + \delta_2,\theta_3) f(\delta_2, \theta_3) \sin 2 \delta_1 &\!=\! 0
\end{align}
and
\begin{align}
    \!\!\int\!\! d\delta_1 d\delta_2 d\theta_3 p_{\theta} (\frac{\pi}{4} \!+\! \delta_1,\frac{\pi}{4} \!+\! \delta_2,\theta_3)  g(\delta_1, \theta_3) \sin 2 \delta_2 &\!=\! 0
\end{align}
\end{subequations}
which is guaranteed by the symmetry property of Assumption~2.

\section{Details on the structure of space-time channels} \label{sec:app_structure}
Here we provide some more technical details about the discussion of \ref{sec:dimension_counting}. We are looking at the Frobenious distance to the boundary of the convex set measured from the completely depolarising map $\mathcal{T}$.

Specifically, we consider super-operators of the form $\mathcal E = \mathcal T + \lambda \mathcal E'$, where $\mathcal E'$ is hermiticity-preserving, maintains 4-way unitality and is orthogonal\footnote{Here we implicitly use that the channels are represented by matrices, and use the Hilbert-Schmidt inner product.} to $\mathcal T$.
We construct $\mathcal E'$ by choosing its unconstrained elements to be random real numbers between minus one and one, and then normalise it using the Frobenius norm so $\left\| \mathcal E' \right\|_2 = 1$.
Increasing the strength $\lambda$, we observe that complete positivity fails at some typical distance, 
$\lambda \approx 0.6\pm 0.2$ for $d=2$, in all of $10000$ random samples\footnote{The observation is true for any $n$-way unitality, $n=1, \ldots,6$.}. 

The Frobenious norm of dual-unitary channels ($q_{DU}=U_{DU} \otimes U_{DU}^*$) is $4$, and its $00,00$ element in the Pauli basis is $1$, therefore the distance from $\mathcal{T}$ to the dual-unitary channels is $\sqrt{15}\approx 3.8$, which is considerably more than in the random direction.

Said differently, starting from the completely depolarising channel, the directions towards the dual-unitary channels are very fine-tuned. In those directions complete positivity fails considerably later than in random directions.

\end{document}